\documentclass[conference]{IEEEtran}
\usepackage[boxruled]{algorithm2e}
\usepackage{cite}
\usepackage{graphicx}
\usepackage{psfrag}
\usepackage{subfig}
\usepackage{url}
\usepackage{amsmath}
\usepackage{array}
\usepackage{amssymb}
\usepackage{amsfonts}
\usepackage{graphicx}
\usepackage{epstopdf}

\newtheorem{lemma}{Lemma}

\newtheorem{definition}{Definition}

\newtheorem{example}{Example}

\title{An Adaptive Modulation Scheme for Two-user Fading MAC with Quantized Fade State Feedback}
\begin{document}

\author{
\authorblockN{Sudipta Kundu and B. Sundar Rajan }
\authorblockA{Dept. of ECE, IISc, Bangalore 560012, India, Email: {$\lbrace$sudiptak, bsrajan$\rbrace$} @ece.iisc.ernet.in
}
}

\maketitle
\thispagestyle{empty}	
%%%%%%%%
\begin{abstract}
With no CSI at the users, transmission over the two-user Gaussian Multiple Access Channel with fading and finite constellation at the input, is not efficient because error rates will be high when the channel conditions are poor. However, perfect CSI at the users is an unrealistic assumption in the wireless scenario, as it would involve massive feedback overheads. In this paper we propose a scheme which uses only quantized knowledge of CSI at the transmitters with the overhead being nominal. The users rotate their constellation without varying their transmit power to adapt to the existing channel conditions, in order to meet certain pre-determined  minimum Euclidean distance requirement in the equivalent constellation at the destination. The optimal modulation scheme has been described for the case when both the users use symmetric $M$-PSK constellations at the input, where $ M=2^\lambda $, $ \lambda $ being a positive integer. The strategy has been illustrated by considering examples where both users use QPSK or 8-PSK signal sets at the input. It is shown that the proposed scheme has better throughput and error performance compared to the conventional non-adaptive scheme, at the cost of a feedback overhead of just $\left\lceil \log _2 \left(\frac{M^2}{8}-\frac{M}{4}+2\right)\right\rceil + 1 $ bits, for the $M$-PSK case.  
\end{abstract}

\section{INTRODUCTION}

A multiple access channel (MAC) consists of multiple users transmitting independent information to a common destination. There is no cooperation among the users. The capacity region for a discrete memoryless MAC is well known \cite{Cover} \cite{Gallager}. 
For a two-user MAC with additive white Gaussian noise (AWGN) the capacity achieving input is the continuous Gaussian alphabet. The two-user Gaussian MAC with finite input constellations like $M$-QAM, $M$-PSK was studied in \cite{HarshanCR} \cite{HarshanCPA}. It was shown that relative rotation between input constellations \cite{HarshanCR}, or a constellation power allocation scheme \cite{HarshanCPA} may be employed to maximize the constellation constrained (CC) capacity regions. Trellis based coding schemes were also suggested to achieve any rate pair within the CC capacity region.

In this paper, a two-user MAC with quasi-static fading is considered, as shown in Fig. \ref{fig:mac}. The two users transmit information to a common destination. The random variables $h_1$ and $h_2$ are  the channel gains for User-1 and User-2 respectively and $h_1$,$h_2$ $\sim \mathcal{CN}(0,1)$, where $ \mathcal{CN}(0,s) $ denotes the circular symmetric complex Gaussian random variable with variance $ s $.  AWGN $ z $ gets added to the received signal at the destination, $z$ $\sim \mathcal{CN}(0,\sigma^2)$. User-$i$ transmits a symbol $ x_i $ from a complex finite constellation $ \mathcal{S}_i $ (like $M$-QAM or $M$-PSK) of unit average energy, \emph{i.e}, $\mathbb{E} [\mid x_i \mid ^2]=1$. Let $ P $ be the average power constraint for each user. The received signal at the destination is thus represented by 
\begin{align}
\nonumber
y=\sqrt{P}h_1x_1+\sqrt{P}h_2x_2+z.
\end{align}
We assume that perfect CSI \emph{i.e.} the tuple $ (h_1,h_2) $ is available only at the destination.  \begin{figure}[t]
\centering
\includegraphics[totalheight=1.5in,width=3in]{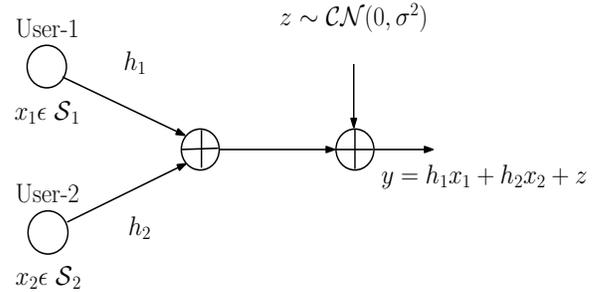}
\caption{Two-user fading MAC with Gaussian noise}	
\label{fig:mac}	
\end{figure}
At the destination the system can be viewed as a single user AWGN channel with the symbols drawn from a sum constellation 
\begin{align}
\nonumber
\mathcal{S}_{\text{sum}} &= \sqrt{P} h_1\mathcal{S}_1+ \sqrt{P} h_2\mathcal{S}_2\\
\nonumber
&= \sqrt{P} h_1(\mathcal{S}_1+\dfrac{h_2}{h_1}\mathcal{S}_2)\\
\label{eqneff}
&= \sqrt{P} h_1\underbrace{(\mathcal{S}_1+\gamma e^{j\theta}\mathcal{S}_2)}_{\mbox{$\mathcal{S}_{\text{eff}}$}},
\end{align}
where $\gamma= | \frac{h_2}{h_1}|$, $\theta = \angle \frac{h_2}{h_1}$ and $\mathcal{S}_{\text{eff}}$ denotes the effective constellation.

\vspace{0.2cm}Without loss of generality it can be assumed that $ \gamma \geq 1 $, as destination has knowledge of both $h_1$ and $h_2$ separately. If $\vert \frac{h_2}{h_1}\vert < 1$, then at the destination the ratio can be simply reversed to compute  $\frac{h_1}{h_2}$. Which one among the two ratios is calculated is made known to the users via a single bit of feedback. For the rest of the paper, we assume that the ratio $\frac{h_2}{h_1}$ is calculated at the destination. However, the results obtained still hold when the ratio calculated is $\frac{h_1}{h_2}$, by interchanging the roles of User-1 and User-2. For the rest of the paper a $M$-PSK constellation refers to a symmetric PSK signal set, with $ M=2^\lambda $, $ \lambda $ being a positive integer. The points in the  $ M $-PSK signal set are of the form $ e^{j\frac{(k-1)2\pi}{M}} $, where $ 1\leq k \leq M $. We assume that $ \mathcal{S}_1 = \mathcal{S}_2= \mathcal{S}$, where $\mathcal{S}$ is an $M$-PSK constellation.  We refer to the pair $ (\gamma , \theta)$ to represent $\gamma e^{j \theta}$ and  call it the fade state throughout the paper. We refer to the complex plane that represents $ \gamma e^{j \theta } $ with $ \gamma \geq 1 $ as the $ (\Gamma, \Theta) $ plane.

Perfect channel state information (CSI) is available at the destination only, which quantizes the $ (\Gamma , \Theta) $ plane into finite number of regions. The quantization obtained is similar to the one used for physical layer network coding in \cite{Akino}, which was subsequently derived analytically in \cite{Vijay}.  This quantized knowledge of the fade state is made available to the users to adapt their modulation scheme via rotation of constellations to compensate for the possibly bad channel conditions. MAC with limited channel state information at transmitter (CSIT) has been studied from an information theoretic point of view in \cite{Reza},\cite{Cemal}. In \cite{Wiese}, it was shown that for a two-user discrete memoryless MAC with additional common message, finer CSIT results in increasing the capacity region. To the best of our knowledge, explicit modulation schemes with finite constellations and quantized fade state feedback has not been reported before.

The contributions and organization of this paper are as follows:
\begin{itemize}
\item
A quantization of the $ (\Gamma , \Theta) $ plane is derived, for the case when both users use $M$-PSK constellations at the input. We illustrate the quantization procedure by taking examples of the QPSK and 8-PSK case. (Section \ref{sec:chan_quan})
\item
A modulation scheme is proposed for the users, which adapts according to the quantized feedback about the  fade state that they receive from the destination, in order to satisfy a certain minimum distance guarantee $\delta$ in $\mathcal{S}_{\text{eff}}$ given in \eqref{eqneff}. The fade states which leads to violation of this minimum distance guarantee have been identified. Adaptation involves rotation of the constellation of one user relative to the other, without any change in transmit power, in order to effectively avoid these bad channel conditions. (Section \ref{sec:adapt})
\item
The procedure to obtain the optimal angles for rotation is stated for the $M$-PSK case. The optimal rotation angles are calculated in closed form for the QPSK and 8-PSK case. (Section \ref{sec:optimal})
\item
An upper bound on $\delta$, \textit{i.e.}, the maximum value of the minimum distance in the effective constellation that can be guaranteed, is derived. (Section \ref{sec:upper})
\item
Simulation results  are presented to show the extent to which the proposed strategy outperforms   the conventional transmission scheme without adaptation. (Section \ref{sec:results})
\end{itemize}
 
\section{CHANNEL QUANTIZATION FOR $ M $-PSK SIGNAL SETS}
In this section we obtain a quantization of the $ (\Gamma, \Theta) $ plane into finite number of regions at the destination.  

\subsection{Distance Distribution in the effective constellation} \label{sec:quandist}
%%%%%%%%%%%%%%%%%%%%%%%%%%%%%%%%%%%%%%%%%%%%%%%%%%%%%%%%%%%%%%%%%
Without loss of generality we assume that the average power constraint of each user is $ P=1 $. It is known that the error performance for an AWGN channel is determined by the Euclidean distance distribution of the input constellation. In our case, the distance distribution of  $ \mathcal{S}_{\text{sum}} $ decides the error performance at the destination.  For any value of $  (\gamma,\theta) $,   $d^{'2}_{(s_1,s_2)_{\text{sum}}\leftrightarrow (s'_1, s'_2)_{\text{sum}}}$ denotes the distance  between the two points $(s_1,s_2)_{\text{sum}}$ and $(s'_1,s'_2)_{\text{sum}}$, where $ (s_1,s_2)_{\text{sum}} $, $(s'_1,s'_2)_{\text{sum}}$ $\in \mathcal{S}_{\text{sum}}$ refer to the points $ \sqrt{P}h_1(s_1+\gamma e^{j \theta} s_2) $  and $\sqrt{P}h_1(s'_1+\gamma e^{j \theta} s'_2)$ respectively with $ s_1,s_2,s'_1,s'_2 \in \mathcal{S} $. It is given by
\begin{align}
\nonumber
d ^{'2}_{(s_1,s_2)_{\text{sum}}\leftrightarrow (s'_1, s'_2)_{\text{sum}}} &= P \vert h_1\vert ^2\vert(s_1-s'_1 )+ \gamma e^{j\theta } (s_2-s'_2) \vert ^2 \\
\label{eqn:sumdistance}
&= P \vert h_1 \vert ^2 d^2_{(s_1,s_2)\leftrightarrow (s'_1, s'_2)},
\end{align}
where \eqref{eqn:sumdistance}, $d^2_{(s_1,s_2)\leftrightarrow (s'_1, s'_2)}$ denotes the distance between the points $(s_1,s_2)$ and $(s'_1,s'_2)$, where $ (s_1,s_2), (s'_1,s'_2) $ refers to the points $ s_1+\gamma e^{j \theta} s_2 $ and $s'_1+\gamma e^{j \theta} s'_2$ in $ \mathcal{S}_{\text{eff}}  $. 
Since $ P \vert h_1 \vert ^2  $ simply scales the distances in $ \mathcal{S}_{\text{eff}} $ we can focus only on
\begin{align}
\label{eqn:dist}
d^2_{(s_1,s_2)\leftrightarrow (s'_1, s'_2)}= \vert(s_1-s'_1 )+ \gamma e^{j\theta } (s_2-s'_2) \vert ^2
\end{align}
as the quantity of interest.

It is clear from \eqref{eqn:dist} that for certain values of $ (\gamma, \theta) $ the distance between points $(s_1,s_2)$ and $(s'_1,s'_2)$ in $\mathcal{S}_{\text{eff}}$ reduces to zero, 
\textit{i.e.} if 
\begin{align}
\label{eqn:sing}
 \gamma e^{j \theta} = -\frac{(s_1-s'_1)}{(s_2-s'_2)}
\end{align}
\noindent then $ d^2_{(s_1,s_2) \leftrightarrow (s_1\prime, s_2\prime) }=0 $. These values of $ (\gamma, \theta) $ are called the \emph{singular fade states} \cite{Akino}, \cite{Vijay}. Singular fade states can also be defined as follows:
\vspace{2pt}
\begin{definition}\label{def:sing}
A fade state $ (\gamma, \theta) $ is said to be a \textit{singular fade state} if  $ \vert \mathcal{S}_{\text{eff}} \vert < M^2 $.
\end{definition}

Clearly, $ \gamma e^{j \theta}=0 $ is a singular fade state, for any arbitrary signal set $ \mathcal{S} $. For any input constellation $\mathcal{S}$ the other non-zero singular fade states are obtained using \eqref{eqn:sing}. For a given input constellation $ \mathcal{S} $, let $ \mathcal{H} $ denote the set of all singular fade states.
\vspace{4pt}
\begin{example}
When $\mathcal{S}_1=\mathcal{S}_2=\mathcal{S}$, where $\mathcal{S}$ is a QPSK constellation, then the non-zero singular fade states are at
\begin{align}
\nonumber
\gamma &= \sqrt{2} , \hspace{5pt} \theta = 45^\circ , 135^\circ ,225^\circ , 315^\circ \\
\nonumber
\gamma &= 1, \hspace{13pt} \theta = 0^\circ , 90^\circ , 180^\circ , 270^\circ\\
\nonumber
\gamma &= \frac{1}{\sqrt{2}}, \,\theta = 45^\circ , 135^\circ ,225^\circ , 315^\circ
\end{align}
\end{example}

Since for an AWGN channel the error performance at the destination is dominated by the minimum distance of the input constellation, it is sufficient to study the minimum distance of $ \mathcal{S}_{\text{eff}} $. Also from Definition \ref{def:sing} minimum distance in $ \mathcal{S}_{eff} $ reduces to zero at the singular fade states. 
The following lemma provides an upper bound on the minimum distance of the effective constellation $ \mathcal{S}_{\text{eff}} $.\vspace{5pt}
\begin{lemma}
\label{lem:uppermin}
When both the users use any arbitrary signal set $\mathcal{S}$ (which includes $M$-PSK, $M$-QAM) at the input, then for any fade state $ (\gamma, \theta) $, the minimum distance $ d_{min}(\gamma , \theta) $ between any two points in $ \mathcal{S}_{\text{eff}} $ is upper bounded by the minimum distance in the input constellation $ d_{min} (\mathcal{S}) $.
\end{lemma}

\begin{proof}
From the definition of $d_{min}^2(\gamma ,\theta)$, we have,
{\small
\begin{align}
\nonumber
d_{min}^2(\gamma ,\theta) &=\hspace{-0.5cm}\min_{(s_1,s_2) \neq (s'_1,s'_2) \in \mathcal{S}^2} \vert (s_1-s'_1)+\gamma e^{j \theta} (s_2-s'_2)\vert ^2\\
\nonumber
&\leq \hspace{-0.5cm}\min_{(s_1,s_2) \neq (s'_1,s'_2) \in \mathcal{S}^2} \left\lbrace \vert s_1-s'_1\vert ^2 +\gamma ^2 \vert s_2-s'_2\vert ^2 \right\rbrace 
\end{align}
\begin{align}
\label{Lemma1_1}
\text{Now}\hspace{5pt} d_{min}^2(\gamma ,\theta) &\leq \min_{s_1 \neq s'_1 \in \mathcal{S}}\vert s_1-s'_1\vert ^2 =d_{min}^2(\mathcal{S}).\\
\label{Lemma1_2}
\text{Also} \hspace{5pt} d_{min}^2(\gamma ,\theta) &\leq\min_{s_2 \neq s'_2 \in \mathcal{S}}\gamma ^2\vert s_2-s'_2\vert ^2=\gamma ^2 d_{min}^2 (\mathcal{S}). 
%&\leq \min_{s_2 \neq s'_2 \in \mathcal{S}}\vert (s_2-s'_2)\vert ^2 \hspace{5pt} (\because \gamma \geq 1) \tag{B}
\end{align}
\begin{align}
\nonumber
\text{From \eqref{Lemma1_1} and \eqref{Lemma1_2}, and using the fact that $\gamma \geq 1$, we have} \\
\nonumber
 d_{min}^2(\gamma ,\theta) \leq \min \lbrace d_{min}^2(\mathcal{S}),\gamma ^2 d_{min}^2(\mathcal{S})\rbrace =d_{min}^2(\mathcal{S}). 
\end{align}
}
\end{proof}

\vspace{3pt}
In the following lemma, it is proved that in order to study the distance profile in $ \mathcal{S}_{\text{eff}} $ it is sufficient to consider $\theta \in \left[ 0,\pi /M\right] $ when both users use $M$-PSK signal sets. Distance profiles for other values of $\theta$ can be obtained from $\theta \in [0, \pi/M] $. We use the term wedge $[\theta_1, \theta _2]$ to denote the region $\gamma \geq 1$ and $\theta \in [\theta_1, \theta_2]$ on the $ (\Gamma , \Theta) $ plane. The lines $ \theta= \theta_1 $ and $ \theta= \theta _2 $ for $ \gamma \geq 1 $ and the arc $ \gamma=1 $ for $ \theta \in [\theta _1, \theta _2] $ form the boundary of the wedge $ [\theta _1 , \theta _2] $.
\vspace{5pt}
 
\begin{lemma}
\label{lemmaangle}
To study the distance profile in $ \mathcal{S}_{\text{eff}} $ when both the users use  $M$-PSK constellations, it is sufficient to consider the case $ 0 \leq \theta \leq \pi /M $. All other cases can be obtained from this.
\end{lemma}
\begin{proof}
The proof is in two steps. First we show that the distance profile is a repetitive structure with period $ 2\pi /M $. Next, it is shown that within the wedge $ [0,2\pi /M] $ the distance profile is symmetric about the bisector of this wedge \textit{i.e.}, the  $ \theta = \pi /M $ line. We have from \eqref{eqneff}, 
\begin{align}
\nonumber
\mathcal{S}_{\text{eff}}= \mathcal{S}+ \gamma e^{j \theta } \mathcal{S} .
\end{align}
For any arbitrary value of $ \theta = \frac{k 2 \pi }{M} +\theta ' $ where $k \in \mathbb{Z},\, 0 \leq \theta ' < \frac{2 \pi}{M}$,
\begin{align}
\nonumber
\mathcal{S}_{\text{eff}}&= \mathcal{S}+ \gamma e^{j \theta ' } (\mathcal{S}e^{j \frac{2 k \pi}{M}})\\
\nonumber
&= \mathcal{S}+ \gamma e^{j \theta ' } \mathcal{S}.
\end{align}
The last equality follows from the fact that rotating a $M$-PSK constellation by an integral multiple of $2\pi /M$ does not alter the distance profile of the constellation. Thus, whatever distance profiles for $ \mathcal{S}_{\text{eff}} $ are obtained for the wedge $ [0,2\pi /M] $,  it is exactly repeated for the remaining $M-1$ wedges to cover the entire range of $ \theta $.
\vspace{3pt}

To show that the distance profiles are symmetric about $\theta= \pi /M$, we need to show that 
$\mathcal{S}+\gamma e^ {j (\frac{\pi}{M} + \alpha)} \mathcal{S}$ and $\mathcal{S}+\gamma e^{j (\frac{\pi}{M} -\alpha )}\mathcal{S}$, where $0\leq \alpha \leq \pi /M$, have the same distance profiles. We have
\begin{align}
\nonumber
\mathcal{S}+\gamma e^{j (\frac{\pi}{M} -\alpha )}\mathcal{S}&=\mathcal{S}+\gamma e^{j (\frac{\pi}{M} -\alpha )}(\mathcal{S}e^ {\frac{-2\pi}{M}})\\
\label{Lemma2_1}
&=\mathcal{S}+\gamma e^{-j (\frac{\pi}{M} +\alpha )}\mathcal{S}.
\end{align}
The first equality is because $\mathcal{S}=\mathcal{S}e^{j k 2\pi /M},$ \emph{i.e.} rotating $\mathcal{S}$ by  $2k\pi /M$ gives the same constellation. Thus for $k=-1$, $\mathcal{S} = \mathcal{S}e^ {\frac{-2\pi}{M}}$. Also due to the symmetric nature of $M$-PSK constellation, the distance distribution of the sum constellation depends only on the relative angle of rotation between the input constellations. Thus $\mathcal{S}+\gamma e^ \beta \mathcal{S}$ and  $\mathcal{S}+\gamma e^ {-\beta} \mathcal{S}$ have same distance profiles, for any $\beta \in [0,\pi]$. This together with \eqref{Lemma2_1} proves the second part of the lemma. 
\end{proof}

From Lemma \ref{lemmaangle}, it is clear that when both users use $M$-PSK signal sets, if $ (\gamma',\theta') $ is a singular fade state, then there exists singular fade states at $ (\gamma',\theta' + p \frac{2\pi}{M}) $, where $1\leq p \leq M-1$ because distance distribution in $\mathcal{S}_{\text{eff}}$ is periodic with period $\frac{2\pi}{M}$. The distance distribution of $\mathcal{S}_{\text{eff}}$ is the basis for channel quantization. Also from Lemma \ref{lemmaangle}, it suffices to obtain such a quantization only for the wedge $ [0, \pi /M]$. This can then be reflected along the $\theta = \pi /M$ line, to give the quantization for the wedge $[0, 2\pi /M]$, which when repeated for the remaining $M-1$ wedges will cover the entire $(\Gamma , \Theta)$ plane. 

\subsection{Channel Quantization for the $M$-PSK case} \label{sec:chan_quan}

In this subsection we propose a technique to obtain the quantization of the $(\Gamma, \Theta)$ plane, when both users use $M$-PSK signal sets. From Lemma 1, when both users  use  $M$-PSK signal sets at the input, the minimum distance in $\mathcal{S}_{\text{eff}}$ for any value of $(\gamma,\theta)$, $d_{min}(\gamma,\theta) \leq d_{min}(\mathcal{S}) =\sqrt{2\left[1-\cos (\frac{2 \pi}{M})\right]}$. Now the following lemma gives the number of singular fade states in the wedge $[0,\pi /M]$.
\vspace{3pt}
\begin{figure*}[t]
\centering
\includegraphics[totalheight=3.42in,width=3.45in]{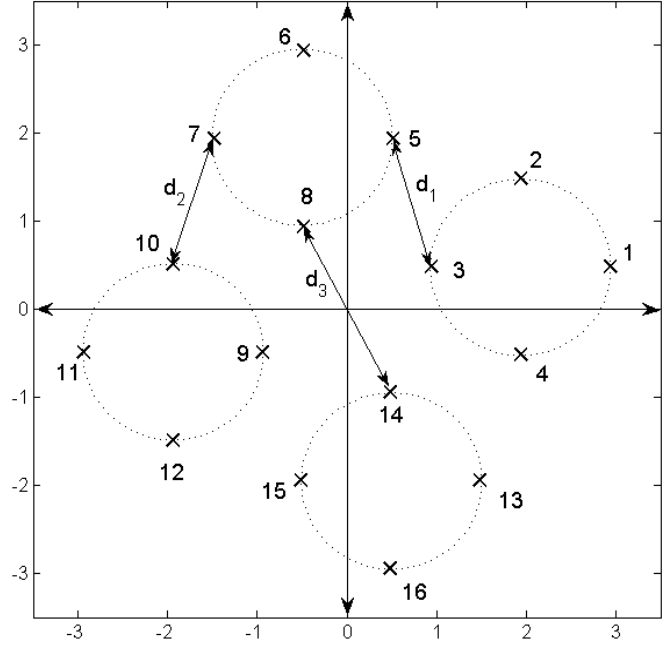}
\caption{$ \mathcal{S}_{\text{eff}} $ constellation, when both users use QPSK signal sets, for $ (\gamma ,\theta )= (2,14^\circ) $. In the figure $ d_i $ stands for the class distance function $ d_{\mathcal{C}_{k_i}} (\gamma, \theta)$. }	
\label{fig:qpskeff}	
\end{figure*}
\begin{lemma}
\label{periodic}
When both users use $M$-PSK signal sets at the input, the number of singular fade states in the wedge $[0, \pi /M]$, is given by $\frac{M^2}{8}-\frac{M}{4}+1$.  Further, these singular fade states lie along the two lines $\theta=0$ and $\theta=\pi /M$.
\end{lemma}
\begin{proof}
From \cite{Vijay}, the total number  of singular fade states other than zero is $\frac{M^3}{4}-\frac{M^2}{2}+M$. Out of these, $M$ lie on the circle $\gamma=1$. It is also known from \cite{Vijay}, that if $\gamma e^{j \theta}$ is a singular fade state, then $\frac{1}{\gamma}e^{-j \theta}$ is also a singular fade state.  Thus,  half of the total number of remaining singular fade states lie inside the circle $\gamma=1$ and the other half lies outside it. This along with the fact that singular fade states are periodic, implies the number of singular fade states for the wedge $ [0,\pi /M]$, is given by
\begin{align}
\nonumber
\frac{1}{M}\left( \frac{\frac{M^3}{4}-\frac{M^2}{2}}{2}+M\right)  =\frac{M^2}{8}-\frac{M}{4}+1.
\end{align}
\noindent Also from \cite{Vijay}, it is clear that these fade states lie along $\theta =0$ and $\theta = \pi /M$ lines.
\end{proof}
We denote this set of all singular fade states lying in the wedge $ [0, \pi/M] $ by $ \mathcal{H}_W $. Let $N_W = \vert \mathcal{H}_W \vert $. 
\vspace{3pt}

Observe from \eqref{eqn:dist}, that the distance between two points in $ \mathcal{S}_{\text{eff}} $ is a function of $ \gamma $ and $ \theta $. Let $ \vert \triangle s_j \vert = \vert s_j-s'_j \vert  $ and $ \phi_j = \angle (s_j-s'_j) $ for $ j=1,\,2 $. Now from \eqref{eqn:dist}, the distance between the elements of the pair $ \lbrace (s_1,s_2),(s'_1,s'_2) \rbrace  \in  \mathcal{S}^2_{\text{eff}} $ is

{\footnotesize
\begin{align}
\label{eq:distclassex}
d^2_{(s_1,s_2)\leftrightarrow (s'_1,s'_2)}= \vert \triangle s_1 \vert ^2 + \gamma ^2 \vert \triangle s_2 \vert ^2 + 2 \gamma \vert \triangle s_1 \triangle s_2 \vert \cos (\theta + \phi_2 -\phi_1 ). 
\end{align}
}
Consider another other pair $ \lbrace (\hat{s}_1,\hat{s}_2), (\hat{s}'_1,\hat{s}'_2)\rbrace \in \mathcal{S}^2_{\text{eff}} $ with $ \hat{s}_1,\, \hat{s}'_1,\,\hat{s}_2,\,\hat{s}'_2 \in \mathcal{S}$, and let $ \vert \triangle \hat{s}_j \vert =  \vert \hat{s}_j - \hat{s}'_j \vert $ and $ \hat{\phi}_j = \angle (\hat{s}_j-\hat{s}'_j) $ for $ j=1,\,2 $. If $ \vert \triangle \hat{s}_1 \vert = \vert \triangle s_1 \vert $, $ \vert \triangle \hat{s}_2 \vert = \vert \triangle  s_2 \vert $ and $ \hat{\phi}_2-\hat{\phi}_1 = \phi_2 - \phi_1 $ or $ \hat{\phi}_2-\hat{\phi}_1 = \pi -(\phi_2 - \phi_1) $, then from \eqref{eq:distclassex}, $ d^2_{(\hat{s}_1,\hat{s}_2)\leftrightarrow (\hat{s}'_1,\hat{s}'_2)} = d^2_{(s_1,s_2)\leftrightarrow (s'_1,s'_2)} $ for all values of $ (\gamma, \theta) $, even though the value of this distance changes with $ (\gamma, \theta) $.
\vspace{3pt}
\begin{definition} \label{def:distclass}
A distance class denoted by $ \mathcal{C} $, is a subset of $ \mathcal{S}_{\text{eff}}^2 $, which contains the pairs of the form $ \lbrace (s_1,s_2),(s'_1,s'_2)\rbrace $, $ (s_1,s_2)\neq (s'_1,s'_2) $ where $ (s_1,s_2)$ and $(s'_1,s'_2) $ denote the complex points in $ \mathcal{S}_{eff} $, such that the distance between the two elements of a pair is same for all pairs in $ \mathcal{C} $ and this property holds for all values of $ \gamma $ and $ \theta $, though the value of the distance depends on $ (\gamma, \theta) $.
\end{definition}

For a given input constellation $ \mathcal{S} $, let $ \bar{\mathcal{C}} $ denote the set of the all distance classes for it. 

\begin{definition} \label{def:classfunc}
Associated with every distance class $ \mathcal{C} $ is a function $ d_{\mathcal{C}}(\gamma ,\theta): (\Gamma,\Theta )\leftrightarrow \mathbb{R} $, called the class distance function, which gives the value of the distance between the two elements of a pair in $ \mathcal{C} $ for any $ (\gamma, \theta) $.
\end{definition}
\begin{definition} \label{def:fadefunc}
For a given fade state $ (\gamma, \theta) $, the function $ d_{\gamma , \theta} (\mathcal{C}): \bar{\mathcal{C}} \rightarrow \mathbb{R} $ gives the value of the distance between the two elements of a pair in $ \mathcal{C} $, for any $ \mathcal{C} \in \bar{\mathcal{C}}$. This is called the fade state distance function.
\end{definition}
\vspace{3pt}

We use integer $ m $, $ 1\leq m \leq M $ to represent the point $ e^{j\frac{(m-1)2 \pi}{M}} $ in $ \mathcal{S} $ \emph{i.e.} the $ M $-PSK signal set. The integer $ q=m+M(n-1) $, $ 1\leq q \leq M^2 $ denotes the complex point in $ \mathcal{S}_{\text{eff}} $ obtained by combining the points $ m $ and $ n $ of $ \mathcal{S} $ \emph{i.e.} it refers to the point $ e^{j\frac{(m-1)2 \pi}{M}} + \gamma e^{j \theta}  e^{j\frac{(n-1)2 \pi}{M}} $ in $ \mathcal{S}_{\text{eff}} $. For each distance class $ \mathcal{C} $, among all the pairs $ (i,j) \in \mathcal{C} $ choose the one with the minimum value of $ i+j $ to be the  representative in $ \mathcal{C} $. If more than one pair  has the same value of $ i+j $ choose the one with the lowest value of $ i $ as the class representative. When the users $ M $-PSK signal set at input, there are $ \frac{M^2(M^2-1)}{2} $ pairwise distances in $ \mathcal{S}_{\text{eff}} $. These pairwise distances are thus partitioned into distance classes.
\begin{example}
Fig. \ref{fig:qpskeff}, shows $ \mathcal{S}_{\text{eff}} $ when both users use QPSK constellations for $(\gamma, \theta)=(2,14^\circ)$. There are 120 pairwise distances and 20 distance classes. These are listed in Table \ref{tab:qpskex}, along with the corresponding class distance functions, and the class representatives.
\end{example}
\begin{table*}[t]
\caption{Distance classes when both users use QPSK signal set.}
\label{tab:qpskex}
\begin{center}
\begin{tabular}{|c|l|l|l|}
\hline
$ k $ & $ \mathcal{C}_k $ & $ d_{\mathcal{C}_k}(\gamma, \theta) $ & Class representative\\ \hline
$ 1 $ & $(1,2),(1,4),(2,3),(3,4),(5,6),(6,7),(7,8),(5,8),$ & $ \sqrt{2} $ & $ (1,2) $ \\
      & $(9,10),(10,11),(11,12),(9,12),(13,14),(14,15),(15,16),(13,16)$ &   & \\ \hline
$ 2 $ & $(1,3),(2,4),(5,7),(6,8),(9,11),(10,12),(13,15),(14,16)$ & $ 2 $ & $(1,3)$  \\ \hline
$ 3 $ & $(1,6),(1,16),(2,15),(4,7),(6,11),(5,12),(11,16),(10,13)$ & $ 2 \gamma ^2 +2+ 4 \gamma \cos \theta $ & $ (1,6) $ \\ \hline
$ 4 $ & $(1,8),(2,7),(2,13),(3,16),(6,9),(7,12),(11,14),(12,13)$ & $ 2\gamma ^2 + 2+ 4 \gamma \sin \theta$ & $ (1,8) $ \\ \hline
$ 5 $ & $ (1,14),(4,15), (3,6),(4,5),(8,11),(6,10),(9,16),(10,15) $ & $2\gamma ^2 + 2- 4 \gamma \sin \theta$ & $ (3,6) $ \\ \hline
$ 6 $ & $ (2,5),(3,8),(3,14),(4,13),(7,10),(8,9),(12,15),(9,14) $ & $ 2 \gamma ^2 +2- 4 \gamma \cos \theta$ & $ (2,5) $ \\ \hline
$ 7 $ & $ (1,5), (1,13),(4,8),(2,6),(2,14),(3,7),(3,15),(4,16),$ & $ 2 \gamma ^2 $  & $ (1,5) $\\ 
& $ (6,10),(7,11),(8,12),(5,9),(12,16),(11,15),(10,14),(9,13) $ & & \\ \hline 
$ 8 $ & $ (1,9),(2,10),(3,11),(4,12),(6,14),(5,13),(7,15),(8,16) $ & $ 4 \gamma ^2  $ & $ (1,9) $ \\ \hline
$ 9 $ & $ (1,7),(2,16),(6,12),(11,13) $ & $ 2 \gamma ^2 +4 + 4 \gamma \cos \theta + 4 \gamma \sin \theta $ & $ (1,7) $ \\ \hline
$ 10 $ & $ (1,15),(4,6),(5,11),(10,16) $ & $ 2 \gamma ^2 +4 + 4 \gamma \cos \theta - 4 \gamma \sin \theta $ & $ (4,6) $ \\ \hline
$ 11 $ & $ (3,13),(2,8),(7,9),(12,14)) $ &  $ 2 \gamma ^2 +4 - 4 \gamma \cos \theta + 4 \gamma \sin \theta $ & $ (2,8) $ \\ \hline
$ 12 $ & $ (3,5),(8,10),(9,15),(4,14) $ & $   2 \gamma ^2 +4 - 4 \gamma \cos \theta - 4 \gamma \sin \theta $ & $ (3,5) $ \\ \hline
$ 13 $ & $ (1,12),(2,11),(6,13),(7,16) $ &  $ 4 \gamma ^2 +2 + 4 \gamma \cos \theta + 4 \gamma \sin \theta $ & $ (1,12) $\\ \hline
$ 14 $ & $ (1,10),(6,15),(4,11),(5,16) $ & $ 4 \gamma ^2 +2 + 4 \gamma \cos \theta - 4 \gamma \sin \theta $ & $ (1,10) $ \\ \hline
$ 15 $ & $ (2,9),(3,12),(7,14),(8,13) $ & $ 4 \gamma ^2 +2 - 4 \gamma \cos \theta + 4 \gamma \sin \theta $ & $ (2,9) $ \\ \hline
$ 16 $ & $ (3,10),(8,15),(4,9),(5,14) $ & $ 4 \gamma ^2 +2 - 4 \gamma \cos \theta - 4 \gamma \sin \theta $ & $ (3,10) $ \\ \hline
$ 17 $ & $ (1,11),(6,16) $ & $ 4 \gamma ^2 +4 + 8 \gamma \cos \theta $ & $ (1,11) $ \\ \hline
$ 18 $ & $ (2,12),(7,13) $ & $ 4 \gamma ^2 +4 + 8 \gamma \sin \theta $ & $ (2,12) $ \\ \hline
$ 19 $ & $ (3,9),(8,14) $ & $ 4 \gamma ^2 +4 - 8 \gamma \cos \theta $ & $ (3,9) $ \\ \hline
$ 20 $ & $ (4,10),(5,15) $ & $ 4 \gamma ^2 +4 - 8 \gamma \sin \theta $ & $ (4,10) $ \\ \hline
\end{tabular}
\end{center}
\end{table*}

We define the set of all class distance function, $ d_{\bar{\mathcal{C}}}(\gamma, \theta) $ and the set of all fade state distance functions $ d_{\Gamma, \Theta} (\mathcal{C}) $ as follows,
\begin{align}
\nonumber
d_{\bar{\mathcal{C}}}(\gamma, \theta) &= \lbrace d_{\mathcal{C}}(\gamma, \theta) \vert \mathcal{C} \in \bar{\mathcal{C}} \rbrace \\
\nonumber
d_{\Gamma, \Theta}(\mathcal{C})&= \lbrace d_{\gamma , \theta}(\mathcal{C})\vert (\gamma, \theta) \in (\Gamma , \Theta)\, \text{plane} \rbrace .
\end{align}

From Definition \ref{def:sing}, at a singular fade state the value of at least one of the class distance functions in $ d_{\bar{\mathcal{C}}}(\gamma, \theta) $ will reduce to zero. 
\begin{lemma}\label{lemma:ordering}
Among the set of all class distance functions that reduce to zero at the singular fade state $ (\gamma',\theta') $, there is a particular one which is the minimum among that set, for all values of $ (\gamma , \theta) \neq (\gamma',\theta' ) $.
\end{lemma}
\begin{proof}
Let $ L $  be the number of class distance functions that reduce to zero at the singular fade state $ (\gamma',\theta') $. Denote these by $ d_{\mathcal{C}_i}(\gamma, \theta) $, $ 1\leq i \leq L $ and let $ \lbrace (s_{1,i}, s_{2,i}),(s'_{1,i}, s'_{2,i})\rbrace $ be the representative element for the distance class $ \mathcal{C}_i $.  From \eqref{eqn:dist} and \eqref{eqn:sing}, we have
\begin{align}
\nonumber
d_{\mathcal{C}_i}^2(\gamma,\theta)&= d^2_{(s_{1,i},s_{2,i})\leftrightarrow (s'_{1,i},s'_{2,i})} \\
\nonumber
&= \vert (s_{1,i}-s'_{1,i}) + \gamma e^{j \theta} (s_{2,i}-s'_{2,i})\vert ^2\\
\nonumber
&= \vert s_{2,i}-s'_{2,i}\vert ^2 \vert \gamma e^{j \theta} + \frac{s_{1,i}-s'_{1,i}}{s_{2,i}-s'_{2,i}}\vert ^2 \\
\label{eqn:ordr1}
&= \vert s_{2,i}-s'_{2,i}\vert ^2  \vert \gamma e^{j \theta } - \gamma' e ^{j \theta'}\vert ^2.
\end{align}
From \eqref{eqn:ordr1}, these $ L $ class distance functions differ only in the constant coefficient $ \vert s_{2,i}-s'_{2,i}\vert ^2  $. From Definition \ref{def:distclass}, all these coefficients are different. Let
\begin{align}
\nonumber
l'= \text{arg} \, \min_{1\leq i \leq L} \vert s_{2,i}-s'_{2,i}\vert ^2.
\end{align}
Now, from \eqref{eqn:ordr1}, $ d_{\mathcal{C}_{l'}}(\gamma, \theta) $ is minimum among all $ d_{\mathcal{C}_i}(\gamma, \theta) $ for all values of $ (\gamma , \theta) \neq (\gamma' , \theta')$. 
\end{proof}
\begin{definition}
\label{rd}
The region corresponding to distance class $ \mathcal{C} $, $ \mathcal{R}(\mathcal{C}) $ denotes the region in the complex plane $ \lbrace (\Gamma, \Theta)/ \mathcal{H} \rbrace $ for which the class distance function $ d_{\mathcal{C}}(\gamma, \theta) $ gives the minimum distance in  $ \mathcal{S}_{\text{eff}}$, \emph{i.e.},
\begin{align}
\nonumber
\mathcal{R}(\mathcal{C}) = \lbrace &(\gamma ,\theta ) \in \lbrace (\Gamma , \Theta) / \mathcal{H} \rbrace \vert \\
\nonumber
&d_{\mathcal{C}}(\gamma ,\theta) \leq d_{\mathcal{C'}}(\gamma , \theta)\: \text{for all}\: \mathcal{C'} \neq \mathcal{C} \in \bar{\mathcal{C}} \rbrace .
\end{align}
\end{definition}
\begin{definition}
When both the users use $ M $-PSK constellations at the input, $ \mathcal{R}_W (\mathcal{C}) $ denotes the portion of the region $ \mathcal{R}(\mathcal{C}) $ lying in the wedge $ [0, \pi/M] $, \emph{i.e.},
\begin{align}
\nonumber
\mathcal{R}_W(\mathcal{C}) = \mathcal{R}(\mathcal{C}) \cap \, \text{wedge}\, [0, \pi/M].
\end{align}
\end{definition}

Note that, when both the users use $ M $-PSK constellations at the input, for some $ \mathcal{C} \in \bar{\mathcal{C}} $ the corresponding region $ \mathcal{R}(\mathcal{C}) $ can be a null set, because the associated class distance function $ d_{\mathcal{C}}(\gamma, \theta) $ does not give the minimum distance in $ \mathcal{S}_{\text{eff}} $ for any value of $ (\gamma, \theta) $ in $ \lbrace (\Gamma, \Theta) \setminus \mathcal{H} \rbrace $. There is always a distance class $\mathcal{C} \in \bar{\mathcal{C}} $ for which the associated class distance function is $ d_{\mathcal{C}(\gamma, \theta)} = d_{min}(\mathcal{S}) $. We denote this particular distance class as $ \mathcal{C}_{d_{min}(\mathcal{S})} $. From Lemma \ref{lem:uppermin}, the value of this class distance function is the upper bound for the minimum distance in $ \mathcal{S}_{\text{eff}} $. For example when both users use QPSK signal sets at the input, then from Table \ref{tab:qpskex}, $ \mathcal{C}_{d_{min}(\mathcal{S})} = \mathcal{C}_1 $ and the associated class distance function is $ d_{\mathcal{C}_1}(\gamma, \theta)= d_{min}(\mathcal{S})= \sqrt{2} $.   
\vspace{3pt}

The procedure to obtain the quantization of the $(\Gamma,\Theta)$ plane, when both users use  $M$-PSK constellations at the input, is as follows:
\begin{enumerate}
\item[Step 1]
Obtain the $N_W$ singular fade states in $ \mathcal{H}_W $ \emph{i.e.}, lying in the wedge $[0, \pi /M]$.  Each of these singular fade state is denoted by $ (\gamma _i,\theta _i) $ where $ 1\leq i\leq N_W $. 
\item[Step 2]
For the singular fade state $ (\gamma_1, \theta _1) $ in $ \mathcal{H}_W $, identify the set of class distance functions in $ d_{\bar{\mathcal{C}}}(\gamma , \theta) $ that reduces to zero at that singular fade state $ (\gamma_1, \theta_1) $. Choose the one among them, which is minimum in that set for all values of $ (\gamma, \theta) \neq (\gamma_1, \theta_1) $. (From Lemma \ref{lemma:ordering}, there is always only one such class distance function.) Let this class distance function be $ d_{\mathcal{C}_{k_1}}(\gamma ,\theta) $ corresponding to distance class $ \mathcal{C}_{k_1} $. Repeat this for all $ (\gamma_i, \theta_i) \in \mathcal{H}_W $, to obtain a set of class distance functions $ \lbrace d_{\mathcal{C}_{k_i}}(\gamma, \theta) $, $ 1 \leq i \leq N_W \rbrace $. Each $ d_{\mathcal{C}_{k_i}}(\gamma, \theta) $ reduces to zero at the singular fade state $ (\gamma_i, \theta_i) $. This is the set of all possible class distance functions other than $ d_{min}(\mathcal{S}) $, that can possibly produce the minimum distance in $ \mathcal{S}_{\text{eff}} $.
\item[Step 3]
To find the region $ \mathcal{R}_W(\mathcal{C}_{k_i}) $, we need to obtain the values of $ (\gamma, \theta) \in \, \text{wedge}\, [0,\pi/M] $ for which  $d^2_{\mathcal{C}_{k_i}}(\gamma, \theta) \leq d^2_{\mathcal{C}_{k_j}}(\gamma, \theta) $ where $1\leq j\neq i \leq N_W$, and $d^2_{\mathcal{C}_{k_i}}(\gamma, \theta) \leq d^2_{min}(\mathcal{S})$. The curves $d^2_{\mathcal{C}_{k_i}}(\gamma, \theta)=d^2_{\mathcal{C}_{k_j}}(\gamma, \theta)$,  $1\leq i \neq j\leq N_W$, form the pairwise boundary between the regions corresponding to the two distance classes $ \mathcal{C}_{k_i} $ and $ \mathcal{C}_{k_j} $. The curves $ d^2_{\mathcal{C}_{k_i}}(\gamma ,\theta) = d^2_{min}(\mathcal{S})$ form the pairwise boundary between the regions corresponding to distance classes $ \mathcal{C}_{k_i} $ and $ \mathcal{C}_{d_{min}(\mathcal{S})} $. The region $\mathcal{R}_W(\mathcal{C}_{k_i})$  is that region in the wedge $ [0,\pi/M] $ excluding the complex point $ (\gamma_i, \theta_i) $, which  is the innermost region bounded by these pairwise boundaries, enclosing the singular fade state $ (\gamma_i,\theta_i) $. For example, Fig. \ref{fig:8pskregion} depicts the region corresponding to the singular fade state at $(1,0)$ when both users use 8-PSK signal sets. In the figure, the curve $d_1^2=d_j^2$ refers to the curve $d_{\mathcal{C}_{k_1}}^2(\gamma,\theta)=d_{\mathcal{C}_{k_k}}^2(\gamma,\theta)$, and $d_1^2= 2-\sqrt{2}$ refers to the curve $d_{\mathcal{C}_{k_1}}^2(\gamma ,\theta) = 2-\sqrt{2}$. It is the innermost region (shaded in the figure) in the wedge $ [0, \pi/8] $ bounded by the pair-wise boundaries $d_{\mathcal{C}_{k_1}}^2(\gamma, \theta)=d_{\mathcal{C}_{k_2}}^2(\gamma, \theta),\, d_{\mathcal{C}_{k_1}}^2(\gamma, \theta)=d_{\mathcal{C}_{k_4}}^2(\gamma, \theta),\, d_{\mathcal{C}_{k_1}}^2(\gamma, \theta)=d_{\mathcal{C}_{k_5}}^2(\gamma, \theta)$ and $ d_{\mathcal{C}_{k_1}}^2(\gamma, \theta)=d_{\mathcal{C}_{k_6}}^2(\gamma, \theta) $ surrounding the singular fade state $ (1,0)$.  
Once the regions $ \mathcal{R}(\mathcal{C}_{k_i}) $, $ 1\leq i \leq N_W $ are obtained, the region exterior to all these regions, lying within the wedge $ [0, \pi /M] $, is the region where $ d_{min}(\mathcal{S})$ is the minimum distance in $ \mathcal{S}_{\text{eff}} $, \emph{i.e.}, the region $ \mathcal{R}_W(\mathcal{C}_{d_{min}(\mathcal{S})}) $. 
\item[Step 4]
The quantization obtained in Step 3, for the wedge $[0,\pi/M]$, can now be extended by the procedure suggested in Section \ref{sec:quandist} to cover the entire $ (\Gamma,\Theta) $ plane.
\end{enumerate}

We will illustrate the procedure with two examples.
\vspace{3pt}

\begin{example} \label{exa:qpskquan} \textit{Channel Quantization for  QPSK signal sets. }\\
\begin{figure}[t]
\centering
\includegraphics[totalheight=2.3in,width=2.25in]{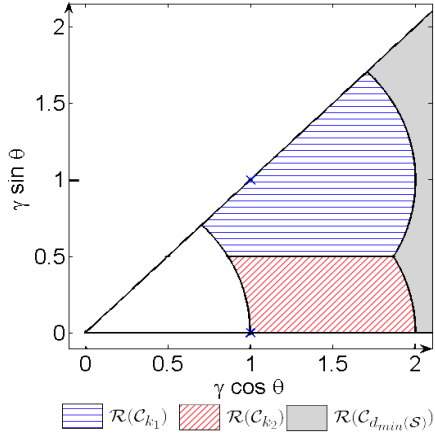}
\caption{Quantization of the wedge $[0,\pi/4]$ for QPSK signal sets. }	
\label{fig:qpsksector}	
\end{figure}
Here we consider the scenario where both users use QPSK constellations at input, \emph{i.e.} $M=4$. From Lemma \ref{periodic}, there are two singular fade states in the wedge  $ [0,\pi/4]$, \emph{i.e.} $N_W=2$ and these are at $(1,0)$ and $(\sqrt{2},\frac{\pi}{4})$.
The class distance functions in $d_{\bar{\mathcal{C}}}(\gamma, \theta)$ which reduce to zero at these singular fade states are  identified. For the singular fade state $(\sqrt{2},\frac{\pi}{4})$ the distance  $d_{\mathcal{C}_{k_1}}$ falls to zero, and for $(1,0)$ the distances $d_{\mathcal{C}_{k_2}}$ and $d_{\mathcal{C}_{k_3}}$ both fall to zero, as shown in Fig \ref{fig:qpskeff}. These are as follows:   
\begin{align}
\nonumber
d_{\mathcal{C}_{k_1}}^2(\gamma, \theta) &= 2\gamma ^2 +4 - 4\gamma \cos \theta - 4 \gamma \sin \theta \\
\nonumber
d_{\mathcal{C}_{k_2}}^2(\gamma, \theta) &= 2\gamma ^2 + 2 -4 \gamma \cos \theta\\
\nonumber
d_{\mathcal{C}_{k_3}}^2(\gamma, \theta) &= 4\gamma ^2 + 4 -8 \gamma \cos \theta = 2d_{\mathcal{C}_{k_2}}^2 \geq d_{\mathcal{C}_{k_2}}^2.
\end{align}
As $d_{\mathcal{C}_{k_3}}(\gamma, \theta) > d_{\mathcal{C}_{k_2}}(\gamma, \theta)$ for all values of $(\gamma,\theta) \neq (1,0)$, we consider the class distance function corresponding to the singular fade state $(1,0)$ as $d_{\mathcal{C}_{k_2}}(\gamma, \theta)$. Now we proceed to obtain the regions $ \mathcal{R}_W(\mathcal{C}_{k_1})$, $ \mathcal{R}_W(\mathcal{C}_{k_2})$ and $\mathcal{R}_W(\mathcal{C}_{d_{min}(\mathcal{S})})$. The pair-wise boundaries are obtained as follows,

{
\footnotesize
\begin{align}
\nonumber
d^2_{\mathcal{C}_{k_1}}(\gamma ,\theta) &= d^2_{\mathcal{C}_{k_2}}(\gamma, \theta) \Rightarrow \gamma \sin \theta = 1/2\\
\nonumber
d^2_{\mathcal{C}_{k_1}} (\gamma, \theta) &= d^2_{min}(\mathcal{S})=2 \Rightarrow (\gamma \cos \theta -1)^2+(\gamma \sin \theta -1)^2 =1\\
\nonumber
d^2_{\mathcal{C}_{k_2}}(\gamma, \theta) &= d^2_{min}(\mathcal{S})=2 \Rightarrow (\gamma \cos \theta -1)^2+\gamma ^2 \sin ^2 \theta  =1.
\end{align}
}
These regions are shown in Fig. \ref{fig:qpsksector} for the wedge $[0,\pi/4] $. This can now be extended to cover the entire range of $ \theta $. For this, the quantization obtained for $ \theta \in [0, \pi/4] $ is reflected  along the line $ \theta = \pi/4 $ to obtain the quantization for the wedge $ [0, \pi/2] $. This is now rotated by $ \frac{k\pi}{2} $, $1\leq k\leq 3  $, to obtain the quantization for $ \theta \in \left[ \frac{k\pi}{2}, \frac{(k+1)\pi}{2} \right]  $, thus covering the entire $ (\Gamma,\Theta) $ plane. This has been shown in Fig. \ref{fig:qpskquan}.
\end{example}    
\vspace{2pt}
\begin{figure}[t]
\centering
\includegraphics[totalheight=2.3in,width=2.2in]{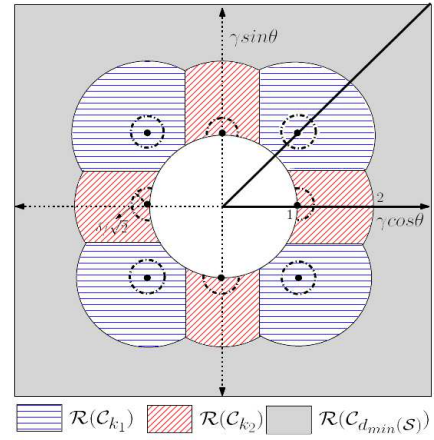}
\caption{Quantization of the entire $(\Gamma, \Theta)$ plane for QPSK signal sets }	
\label{fig:qpskquan}	
\end{figure}

\begin{example} \textit{Channel quantization for 8-PSK signal sets }

\noindent Here we consider the scenario when both users use 8-PSK signal sets at the input. The number of singular fade states lying in the wedge $[0, \pi/8]$ is, $N_W=7$. The singular fade states $(\gamma _i,\theta _i)$, $1\leq i \leq 7$ are as follows:
\begin{align}
\nonumber
&(1, 0),\: (\sqrt{2},0), \:(1+\sqrt{2} ,0), \\
\nonumber
&\left(\sqrt{4-2\sqrt{2}} ,\frac{\pi}{8}\right),\: \left(\sqrt{1+\frac{1}{\sqrt{2}}} ,\frac{\pi}{8}\right),\\
\nonumber
&\left(\sqrt{2+\sqrt{2} } ,\frac{\pi}{8}\right),\: \left(\sqrt{4+2\sqrt{2}} ,\frac{\pi}{8}\right).
\end{align} 
The class distance functions in $d_{\bar{\mathcal{C}}}(\gamma, \theta)$ which reduce to zero for each of the above singular fade states are identified. When more than one class distance function reduces to zero at a singular fade state $ (\gamma_i, \theta_i) $, the one which is minimum among them for any other $ (\gamma, \theta) \neq (\gamma_i, \theta_i) $ is chosen. The class distance functions $ d_{\mathcal{C}_{k_i}}(\gamma, \theta) $, $ 1 \leq i \leq 7 $  are identified as shown in Fig. \ref{fig:8pskeff}. In the figure $ d_i $ denotes the class distance function $ d_{\mathcal{C}_{k_i}}(\gamma , \theta) $. These are as follows

{
\footnotesize
\begin{align}
\nonumber
d^2_{\mathcal{C}_{k_1}}(\gamma ,\theta) &= (2-\sqrt{2})(\gamma ^2 - 2\gamma \cos \theta +1)\\
\nonumber
d^2_{\mathcal{C}_{k_2}}(\gamma ,\theta) &= 2\gamma ^2 + 4 - 4\sqrt{2}\gamma \cos \theta \\
\nonumber
d^2_{\mathcal{C}_{k_3}}(\gamma ,\theta) &= (2-\sqrt{2})\gamma ^2 + 2 +\sqrt{2}- 2\sqrt{2}\gamma \cos \theta \\
\nonumber
d^2_{\mathcal{C}_{k_4}}(\gamma ,\theta) &= (2+\sqrt{2})\gamma ^2 +4 -2(2+\sqrt{2})\gamma \cos \theta -2\sqrt{2}\gamma\sin \theta \\
\nonumber
d^2_{\mathcal{C}_{k_5}}(\gamma ,\theta) &= 2\gamma ^2 +2 + \sqrt{2}-2(1+\sqrt{2})\gamma\cos\theta - 2\gamma\sin\theta \\
\nonumber
d^2_{\mathcal{C}_{k_6}}(\gamma ,\theta) &= (2-\sqrt{2})\gamma ^2 +2 -2(\sqrt{2}-1)\gamma\sin\theta - 2\gamma\cos\theta \\
\nonumber
d^2_{\mathcal{C}_{k_7}}(\gamma ,\theta) &= (2-\sqrt{2})\gamma ^2 +4 -2(2-\sqrt{2})\gamma\sin\theta -2\sqrt{2}\gamma\cos\theta.
\end{align}
}
Now to obtain the region $\mathcal{R}_W(\mathcal{C}_{k_1})$ we need to obtain the curves $d^2_{\mathcal{C}_{k_1}}(\gamma,\theta)=d^2_{\mathcal{C}_{k_j}}(\gamma ,\theta) $, $2\leq j\leq7$ and $d^2_{\mathcal{C}_{k_1}}(\gamma ,\theta) = d^2_{min}(\mathcal{S})= 2-\sqrt{2}$. These are as follows:

{
\scriptsize
\begin{align}
\nonumber
d^2_{\mathcal{C}_{k_1}}(\gamma,\theta) &= d^2_{\mathcal{C}_{k_2}}(\gamma,\theta) \Rightarrow (\gamma\cos\theta -(3-\sqrt{2}))^2+\gamma ^2 \sin ^2 \theta = 10 -7\sqrt{2}\\
\nonumber
d^2_{\mathcal{C}_{k_1}}(\gamma,\theta) &= d^2_{\mathcal{C}_{k_3}}(\gamma , \theta) \Rightarrow \gamma\cos\theta = \frac{1}{ 2-\sqrt{2}}\\
\nonumber
d^2_{\mathcal{C}_{k_1}}(\gamma, \theta) &= d^2_{\mathcal{C}_{k_4}}(\gamma, \theta) \Rightarrow (\gamma\cos\theta -1)^2+(\gamma \sin\theta - \frac{1}{2})^2= \frac{3 -2\sqrt{2}}{4}\\
\nonumber
d^2_{\mathcal{C}_{k_1}}(\gamma,\theta) &= d^2_{\mathcal{C}_{k_5}}(\gamma, \theta) \Rightarrow (\gamma\cos\theta -2+\frac{1}{\sqrt{2}})^2+(\gamma\sin \theta - \frac{1}{\sqrt{2}})^2 \\
\nonumber
& \hspace{2in} = 3 -2\sqrt{2}\\
\nonumber
d^2_{\mathcal{C}_{k_1}}(\gamma ,\theta) &= d^2_{\mathcal{C}_{k_6}}(\gamma,\theta) \Rightarrow \gamma\cos\theta +\gamma \sin \theta = 1 +\frac{1}{\sqrt{2}}\\
\nonumber
d^2_{\mathcal{C}_{k_1}}(\gamma,\theta) &= d^2_{\mathcal{C}_{k_7}}(\gamma, \theta) \Rightarrow \sqrt{2}\gamma\cos\theta +\gamma \sin \theta = \frac{3}{2}+ \sqrt{2}\\
\nonumber
d^2_{\mathcal{C}_{k_1}}(\gamma, \theta) &= 2-\sqrt{2} \Rightarrow (\gamma\cos\theta -1)^2+\gamma ^2 \sin ^2 \theta = 1.
\end{align}
}
\begin{figure}[t]
\centering
\includegraphics[totalheight=2.5in,width=2.5in]{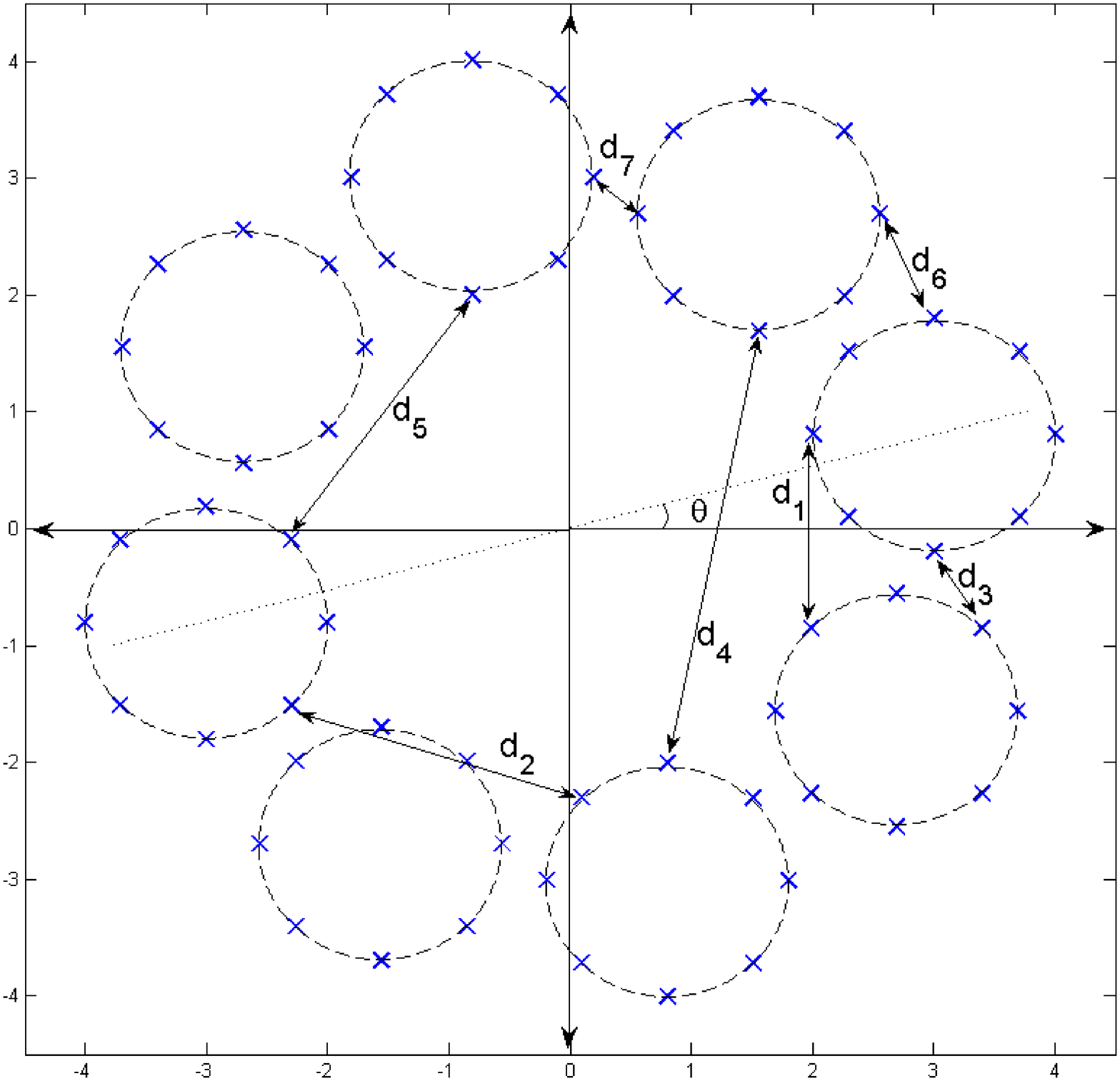}
\caption{Effective constellation $\mathcal{S}_{\text{eff}}$ for $(\gamma,\theta) = (2.9,10^\circ)$ when both users use 8-PSK signal set.}	
\label{fig:8pskeff}	
\end{figure}

All of the above curves are shown in Fig. \ref{fig:8pskregion}. In the figure, the curve $d_1^2=d_j^2$ refers to the curve $d_{\mathcal{C}_{k_1}}^2(\gamma,\theta)=d_{\mathcal{C}_{k_k}}^2(\gamma,\theta)$, and $d_1^2= 2-\sqrt{2}$ refers to the curve $d_{\mathcal{C}_{k_1}}^2(\gamma ,\theta) = 2-\sqrt{2}$. The region $\mathcal{R}_W(\mathcal{C}_{k_1})$ is the innermost region in the wedge $ [0,\pi/8] $  bounded by these curves, surrounding the point $ (1,0) $. It is the shaded  region in the Fig. \ref{fig:8pskregion}. All the regions $\mathcal{R}_W(\mathcal{C}_{k_i})$, $1\leq i \leq 7$ can be obtained by the same procedure. The region exterior to all these regions, inside the wedge $[0, \pi/8] $ is the region $\mathcal{R}_W(\mathcal{C}_{d_{min}(\mathcal{S})})$, where $d_{min}(\mathcal{S}) = \sqrt{2-\sqrt{2}}$ is the minimum distance in $\mathcal{S}_{\text{eff}}$. Thus, the quantization of the fade states for  the wedge $[0, \pi /8]$ is obtained. It is shown is Fig. \ref{fig:8psksector} in the next page. This can now be extended to cover the entire $ (\Gamma, \Theta) $ plane by the similar technique used for the QPSK case in the Example \ref{exa:qpskquan}. 
\end{example}

\section{THE ADAPTIVE MODULATION SCHEME}
In this section the fade states which results in reducing the minimum distance in $ \mathcal{S}_{\text{eff}} $ below a minimum distance guarantee of $ \delta $ are identified. Then a modulation scheme is proposed for the users to avoid these fade states by suitable relative rotation between the signal sets used by the two users. 
 
\subsection{Adaptive Modulation Scheme} \label{sec:adapt}

It is clear from Section \ref{sec:chan_quan} that the minimum distance in $ \mathcal{S}_{\text{eff}} $ falls to zero at the singular fade states. For fade states $ (\gamma ,\theta) $ lying close to a singular fade state, the minimum distance in $ \mathcal{S}_{\text{eff}} $ is very low, resulting in degradation of error performance at the destination. Hence, such values of $ (\gamma ,\theta)$ have to be avoided to provide better performance.

Our goal is to provide a minimum distance guarantee of $ \delta $ in $ \mathcal{S}_{\text{eff}} $ \emph{i.e.}, not allow minimum distance in $ \mathcal{S}_{\text{eff}} $ to fall below $ \delta $. In the previous subsection, the regions $ \mathcal{R}(\mathcal{C}_{k_i}) $ on the $ (\Gamma , \Theta) $ plane was identified, in which the class distance function  $ d_{\mathcal{C}_{k_i}}(\gamma, \theta) $,  $ 1\leq i \leq N_W $  gives the minimum distance in $ \mathcal{S}_{\text{eff}} $. In order to satisfy the minimum distance guarantee of $ \delta $ in $ \mathcal{S}_{\text{eff}} $, it is thus required to avoid the fade states $ (\gamma,\theta) $ for which $ d_{\mathcal{C}_{k_i}}(\gamma ,\theta) < \delta $, $ 1\leq i \leq N_W $.

At the singular fade state $ (\gamma_i, \theta _i) $, the class distance function $ d_{\mathcal{C}_{k_i}}(\gamma, \theta) $ reduces to zero. If $ \lbrace (s_{1,i},s_{2,i}), (s'_{1,i},s'_{2,i})\rbrace $ is the representative element for the distance class $\mathcal{C}_{k_i}$, then at the singular fade state  $ (\gamma _i, \theta _i) $ the two points $ (s_{1,i},s_{2,i}) $ and $ (s'_{1,i},s'_{2,i}) $ collapse to a single point in $ \mathcal{S}_{\text{eff}} $.
From the definition of singular fade state \eqref{eqn:sing}, we have 
\begin{align}
\nonumber
\gamma _i e^{j \theta _i} = -\frac{s_{1,i}-s'_{1,i}}{s_{2,i}-s'_{2,i}}.
\end{align} 
From  \eqref{eqn:dist}, we have
\begin{align}
\label{eqn:di-eff}
d_{\mathcal{C}_{k_i}}(\gamma , \theta) = \vert s_{2,i}-s'_{2,i}\vert \vert \gamma e^{j \theta } - \gamma _i e ^{j \theta _i}\vert.
\end{align}
Only those fade states $ (\gamma ,\theta) $ which results in  $ d_{\mathcal{C}_{k_i}}(\gamma, \theta) < \delta $, $1\leq i \leq N_W$ have to be avoided, \emph{i.e.}, from \eqref{eqn:di-eff}, we need to avoid the fade states $ (\gamma ,\theta)  $, for $ 1\leq i \leq N_W $, where 
\begin{align}
\label{eqn:violation}
\vert \gamma e^{j\theta} - \gamma _i e^{j \theta _i}\vert < \frac{\delta}{\vert s_{2,i}-s'_{2,i}\vert}.
\end{align}
\noindent The above equation represents a circular region in the complex plane $ (\Gamma, \Theta) $ centred at the singular fade state $ (\gamma _i, \theta _i) $ and radius $ \delta /\vert s_{2,i}-s'_{2,i}\vert $. 
We call these circular regions, the \textit{violation circles} because when the fade state lies inside them the minimum distance requirement of $ \mathcal{S}_{\text{eff}} $ is violated. For fade states outside the violation circles the minimum distance in $ \mathcal{S}_{\text{eff}} $ is always greater than $ \delta $. The radius of the violation circle centred at the singular fade state $ (\gamma_i, \theta _i) $ by is denoted by $ \rho (\gamma_i, \theta_i) $. Formally violation circles are defined as follows:
\vspace{3pt}

\begin{definition}
\emph{Violation Circles} are circular regions on the $(\Gamma ,\Theta)$ plane with centres at the singular fade states $ (\gamma _i,\theta _i)= -\frac{s_{1,i}-s'_{1,i}}{s_{2,i}-s'_{2,i}}$ with  radius $ \rho(\gamma_i,\theta_i)=\delta /\vert s_{2,i}-s'_{2,i}\vert $ for $1\leq i \leq N_W$.
\end{definition}
\vspace{3pt}

It is observed that the the violation circles centred at $ (\gamma_i,\theta_i+p\frac{2\pi}{M})  $, $ 1\leq i  \leq N_W \, \text{and} \, 1\leq p \leq M-1 $ have the same radius as the one centred at $ (\gamma_i,\theta_i) $, where $ (\gamma_i , \theta_i) \in \mathcal{H}_W $. This is because from Lemma \ref{lemmaangle}, the corresponding effective constellations are the same.  
\begin{figure*}[t]
\centering
\includegraphics[totalheight=2.6in,width=5in]{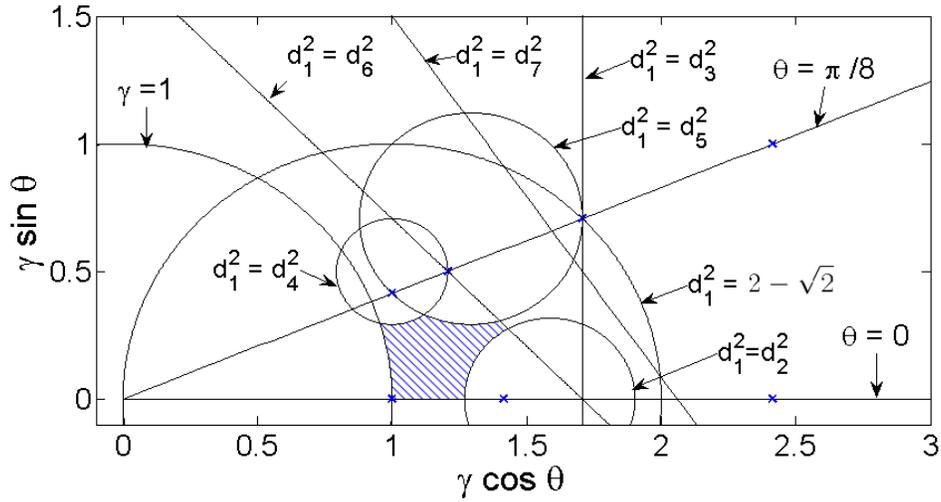}
\caption{The figure explains how to region $\mathcal{R}(d_1)$ (shaded region in the figure) corresponding to the singular fade state $(1,0^\circ)$ is obtained.}	
\label{fig:8pskregion}	
\end{figure*}
\begin{example}

\begin{figure*}[!t]
\centering
\includegraphics[totalheight=3.8in,width=6in]{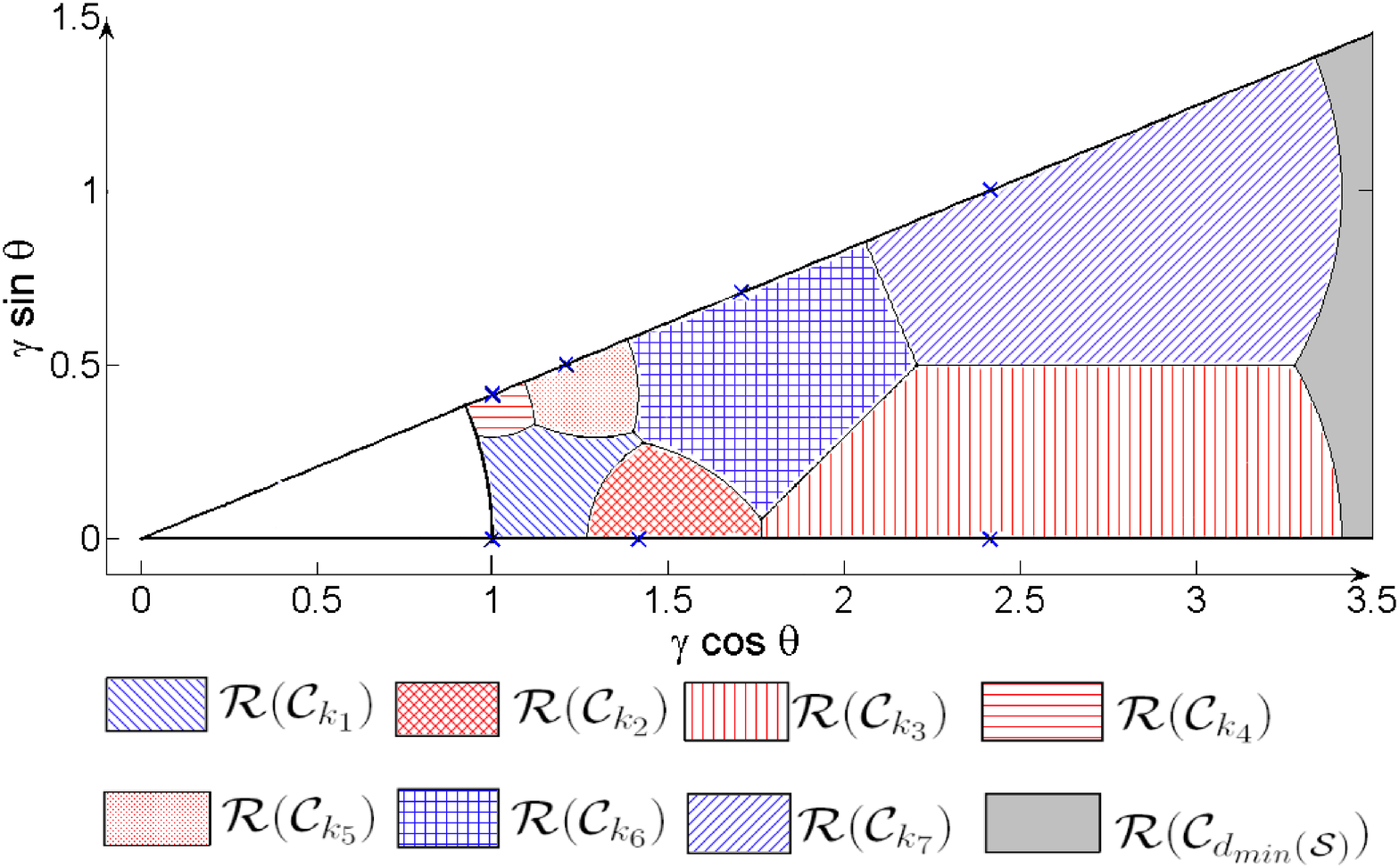}
\caption{Channel quantization for $\theta \in [0,\pi/8]$ when both users use 8-PSK signal sets}	
\label{fig:8psksector}	
\end{figure*}

When both users use QPSK constellations at the input, the violation circles are as follows. At singular fade state $ (\sqrt{2},\pi /4) $ the class distance function $ d_{\mathcal{C}_{k_1}}(\gamma,\theta) $ reduces to zero, \emph{i.e.} the points $ 3 $ and $ 5 $ in $ \mathcal{S}_{\text{eff}} $ collapse to a single point, as shown in Fig. \ref{fig:qpskeff}. The $ 3 $ and $ 5 $ are obtained after combining $ 3,\,1 \in \mathcal{S} $ and $ 1,\,2 \in \mathcal{S} $ respectively.  Thus, $ \vert s_{2,1}-s'_{2,1} \vert = \sqrt{2}$. The violation circle corresponding to class distance function $ d_{\mathcal{C}_{k_1}}(\gamma,\theta) $ is the circular region centred at $ (\sqrt{2}, \pi/4) $ and radius $\delta/\sqrt{2}$. Similarly, for the singular fade state at $ (1,0) $ the violation circle has a radius of $\delta/\sqrt{2}$. These are shown by dotted circles around the singular fade states in Fig. \ref{fig:qpskquan}. 
\end{example}

\begin{example}
When both users are using 8-PSK constellations at the input, the centres and radii of the violation circles corresponding to the distances $ d_{\mathcal{C}_{k_i}},\, 1\leq i \leq 7 $ are tabulated in  Table \ref{table1}.
\begin{table}[t]
\caption{Centre and Radius of violation circles for 8-PSK case}
\label{table1}
\begin{center}
\begin{tabular}{|c|l|c|}
\hline
$ i $ & Centre & Radius \\ \hline
$ 1 $ & $(1,0)$ & $\frac{\delta}{\sqrt{2-\sqrt{2}}}$  \\ \hline
$ 2 $ & $(\sqrt{2},0)$ & $\frac{\delta}{\sqrt{2}}$  \\ \hline
$ 3 $ & $(\sqrt{2}+1 ,0)$ & $\frac{\delta}{\sqrt{2-\sqrt{2}}}$ \\ \hline
$ 4 $ & $\left(\sqrt{4-2\sqrt{2}},\frac{\pi}{8}\right)$ & $\frac{\delta}{\sqrt{2+\sqrt{2}}}$ \\ \hline
$ 5 $ & $\left(\sqrt{1+\frac{1}{\sqrt{2}}} ,\frac{\pi}{8}\right)$ & $\frac{\delta}{\sqrt{2}}$ \\ \hline
$ 6 $ & $\left(\sqrt{2+\sqrt{2}} ,\frac{\pi}{8}\right)$ & $\frac{\delta}{\sqrt{2-\sqrt{2}}}$ \\ \hline
$ 7 $ & $\left(\sqrt{4+2\sqrt{2}},\frac{\pi}{8}\right)$ & $\frac{\delta}{\sqrt{2-\sqrt{2}}}$ \\ \hline
\end{tabular}
\end{center}
\end{table}
\end{example}

When the fade state $ (\gamma ,\theta) $ lies inside any of the violation circles the users need to adapt their transmission, in order to avoid these fade states effectively. One way to achieve this without increasing the transmit power, is to rotate the constellation of User-2. Rotation can be interpreted as simply altering the phase of the fade state.
\begin{lemma}
\label{relative}
When the fade state is $(\gamma , \theta )$, rotation of the constellation of User-2 by an angle $\alpha$ with respect to the constellation of User-1 in an anticlockwise direction, results in effectively altering the phase of the fade state from $ \theta $ to $ \theta + \alpha $.
\end{lemma}
\begin{proof}
Let $\mathcal{S}$ be the constellation being used by both the users at the input. Now the User-2 rotates its constellation by angle $ \alpha $ in the anticlockwise direction, such that it can  now  be represented as $ e^{j \alpha} \mathcal{S} $. The effective constellation $ \mathcal{S}_{\text{eff}} $ can be written as 
\begin{align}
\nonumber
\mathcal{S}+ \gamma e^ {j \theta}\lbrace e^{j \alpha} \mathcal{S}\rbrace =\mathcal{S}+ \gamma e^ {j (\theta +\alpha)} \mathcal{S}.
\end{align} 
Hence, the fade state $(\gamma , \theta)$ is transformed to $(\gamma, \theta + \alpha)$ after rotation.
\end{proof}

The proposed strategy is thus, to rotate the signal set of User-2 whenever the fade state $(\gamma ,\theta)$ lies within any of the violation circles such that the transformed fade state lies outside the violation circles, in order to satisfy the minimum distance guarantee in $ \mathcal{S}_{\text{eff}} $. For fade states outside the violation circles no rotation is required. The destination which has complete CSI sends feedback of $\lceil \log _2(N_W+1)\rceil$ bits to the users to indicate in which one of the violation circles the fade state lies, or if it lies outside all of them.  

\subsection{Optimal Angle of Rotation for the $M$-PSK case} \label{sec:optimal}
\begin{definition}
An optimal rotation angle, for a violation circle with centre at singular fade state $(\gamma_i,\theta_i)$ , $ 1\leq i \leq N_W $ is that angle of rotation which maximizes the minimum distance in $\mathcal{S}_{\text{eff}}$ for the same transmit power, when fade state $ (\gamma ,\theta ) = (\gamma_i, \theta _i)$.
\end{definition}

It should be noted for non-singular fade states inside the violation circle, the minimum distance in $ \mathcal{S}_{\text{eff}} $ after rotation will be less than what it could have been if the fade state $ (\gamma ,\theta) $ was exactly known at the users. When both users use $M$-PSK constellations at the input, it is sufficient to concentrate on the violation circles centred at $(\gamma_i, \theta_i) \in \mathcal{H}_W $ , because the optimal rotation angles for these, is also optimal for other such circles centred at $(\gamma_i, \theta _i  +p \frac{2 \pi}{M})$,  where $ 1 \leq p \leq M-1$. This follows from the fact that the corresponding effective constellations are equivalent.

From Lemma \ref{relative}, rotation of the constellation of User-2 relative to User-1, results in effectively altering the phase of the fade state. Rotation thus results in moving the violation circle with centre at singular fade state $(\gamma _i, \theta _i) $, along a circular arc such that its centre always lie on the curve $\gamma=\gamma_i$ and within the wedge $ [0, \pi/M] $. In order to obtain the optimal angle of rotation for the violation circle centred at $(\gamma_i, \theta _i)$, it is required to calculate for fixed $\gamma=\gamma_i$ the value of phase $\theta = \theta_{i,opt} $, $ \theta \in [0,\pi/M] $ which maximizes the minimum distance in $\mathcal{S}_{\text{eff}}$. Thus after rotation the violation circle centred $ (\gamma_i ,\theta _i) $ is shifted such that its new centre is the point $ (\gamma_i, \theta_{i,opt}) $ on the $ (\Gamma, \Theta) $ plane. We refer to this rotated violation circle as the effective shifted circle. We now prove two lemmas to obtain the value $ \theta_{i,opt} $ when both users use $M$-PSK signal sets.

\begin{lemma}
\label{increasedi}
Let $ d_{\mathcal{C}_{k_i}}(\gamma, \theta) $ be the class distance function which reduces to zero at the singular fade state $ (\gamma_i, \theta_i) \in \mathcal{H}_W $. For a fixed $ \gamma_0 $ and $ (\gamma_0, \theta) $ lying within the wedge $ [0, \pi/M] $, the value of $ d_{\mathcal{C}_{k_i}}(\gamma_0,\theta) $ increases as the difference $ \vert \theta_i - \theta \vert $ increases.
%For fixed $ \gamma=\gamma_0 $ and $ \theta \in [0,\pi/M] $, the value of the class distance function $ d_{\mathcal{C}_{k_i}}(\gamma, \theta) $, which reduces to zero at the singular fade state $ (\gamma_i,\theta_i) $, increases as the difference $ \vert \theta_i -\theta \vert $ increases.
\end{lemma}
\begin{proof}
From \eqref{eqn:di-eff}, we have 
{
\small
\begin{align}
\nonumber
d_{\mathcal{C}_{k_i}}^2(\gamma_0, \theta) &= \vert s_{2,i}-s'_{2,i}\vert ^2  \vert \gamma _0 e^{j\theta} - \gamma _i e^{j \theta _i}\vert ^2\\
\nonumber
&= \vert s_{2,i} - s'_{2,i}\vert ^2 \gamma _0 ^2 \vert  1 -\gamma' e^{j(\theta_i - \theta )}\vert ^2 \,\text{where}\, \gamma'=\gamma_i /\gamma _0 >0\\
\nonumber
&= \vert s_{2,i} - s'_{2,i}\vert ^2 \gamma _0 ^2 \left\lbrace 1+ {\gamma '}^2 - 2\gamma ' \cos \vert \theta _i - \theta \vert \right\rbrace. 
\end{align}
}
As $ \cos \phi $ is a decreasing function of $\phi$, $ 0\leq \phi \leq \pi/M $, $ d_{\mathcal{C}_{k_i}}^2(\gamma_0, \theta) $ increases as $ \vert \theta_i - \theta \vert $ increases. This proves the lemma.  
\end{proof}
%As one move away from the singular fade state $ (\gamma _i, \theta _i) $ along the arc $ \gamma = \gamma _i $ and $ \theta \in [0, \pi/M] $, the minimum distance in $ \mathcal{S}_{\text{eff}} $ increases, as 
\begin{lemma}
\label{lem:optimalboundary}
Let $(\gamma_i, \theta_i)$ be a singular fade state. Let $\mathcal{R}_W(\mathcal{C}_{k_a})$ and $\mathcal{R}_W(\mathcal{C}_{k_b})$ be the regions surrounding the singular fade states $(\gamma_a, \theta_a)$ and $(\gamma_b, \theta_b)$ respectively. Consider the arc traced by the point $(\gamma_i, \theta)$ that lies within the wedge $[0,\pi/M]$ as $\theta$ varies in the direction to move away from the singular fade state $(\gamma_i, \theta_i)$. Let the region $\mathcal{R}(\mathcal{C}_{k_a})$ is encountered before $\mathcal{R}(\mathcal{C}_{k_b})$ as $\theta$ varies. We have the following:
\begin{itemize}
\item[(i)]
The minimum distance in $\mathcal{S}_{\text{eff}}$ is maximized at one of the points of intersection of this arc and the boundaries between the regions $\mathcal{R}(\mathcal{C}_{k_a})$ and $\mathcal{R}(\mathcal{C}_{k_b})$.
\item[(ii)]
Among all the points in (i), those that lie on the boundary between $\mathcal{R}(\mathcal{C}_{k_a})$ and $\mathcal{R}(\mathcal{C}_{k_b})$ with $\theta_i = \theta_a = \theta_b$ can never correspond to the maximum value. 
\end{itemize}
\end{lemma}
\begin{proof}
The proof for part (i) is as follows. Let $\mathbf{A}$ and $\mathbf{B}$ be the first and the second points of intersection of the arc with the boundary of the region $\mathcal{R}_W(\mathcal{C}_{k_a})$. From Section II-B, the value of class distance function $d_{\mathcal{C}_{k_a}}(\gamma, \theta)$ gives the minimum distance in $\mathcal{S}_{\text{eff}}$ when $ (\gamma, \theta) \in \mathcal{R}_W({\mathcal{C}_{k_a}})$.  Now if $\theta_a = \theta_i$, then from Lemma \ref{increasedi}, $d_{\mathcal{C}_{k_a}}(\gamma_i, \theta)$ increases as $(\gamma_i, \theta)$ moves from $\mathbf{A}$ to $\mathbf{B}$. Likewise, if $\theta_a \neq \theta_i$, then again from Lemma \ref{increasedi}, $d_{\mathcal{C}_{k_a}}(\gamma_i, \theta)$ decreases as one moves from $\mathbf{A}$ to $\mathbf{B}$. Thus the minimum distance in $\mathcal{S}_{\text{eff}}$ can never be maximum for $(\gamma_i, \theta)$ lying inside the regions $\mathcal{R}_W(\mathcal{C}_{k_a})$. It can only be maximized at the points of intersection of the arc with the boundary of the region.
 
To prove the second part of the lemma, we assume $\theta_i= \pi/M$. Consider the scenario shown in Fig. 8. To prove part (ii), we need to show that the minimum distance in $\mathcal{S}_{\text{eff}}$ can never be maximum at the point $\mathbf{B}$. In the region $\mathcal{R}_W(\mathcal{C}_{k_a})$ the value of class distance function $d_{\mathcal{C}_{k_a}}(\gamma, \theta)$ gives the minimum distance in $\mathcal{S}_{\text{eff}}$. As $(\gamma_i,\theta)$ moves from $\mathbf{A}$ to $\mathbf{B}$, from Lemma \ref{increasedi}, the minimum distance in $\mathcal{S}_{\text{eff}}$, \emph{i.e.} the value of $d_{\mathcal{C}_{k_a}}(\gamma_i, \theta)$, increases  since $\vert \theta_a- \theta \vert$ increases. At $\mathbf{B}$, we have $d_{\mathcal{C}_{k_a}}(\gamma_i, \theta)= d_{\mathcal{C}_{k_b}}(\gamma_i, \theta)$. Beyond $\mathbf{B}$, in the region $\mathcal{R}_W(\mathcal{C}_{k_b})$, the class distance function $d_{\mathcal{C}_{k_b}}(\gamma, \theta)$ gives the minimum distance in $\mathcal{S}_{\text{eff}}$. As $(\gamma_i, \theta)$ moves from $\mathbf{B}$ to $\mathbf{C}$, $\vert \theta_b - \theta \vert $ increases, thus from Lemma \ref{increasedi}, the minimum distance in $\mathcal{S}_{\text{eff}}$, \emph{i.e.} the value of $d_{\mathcal{C}_{k_b}}(\gamma_i, \theta)$ continues to increase. So $\mathbf{B}$ can never correspond to the point where the minimum distance in $\mathcal{S}_{\text{eff}}$ is maximized.
The proof for the case  when $\theta_i=0$  is exactly similar to the above proof.
This completes the proof.
\end{proof}
\begin{figure} [t]
\centering
\includegraphics[totalheight=1.9in,width=3.2in]{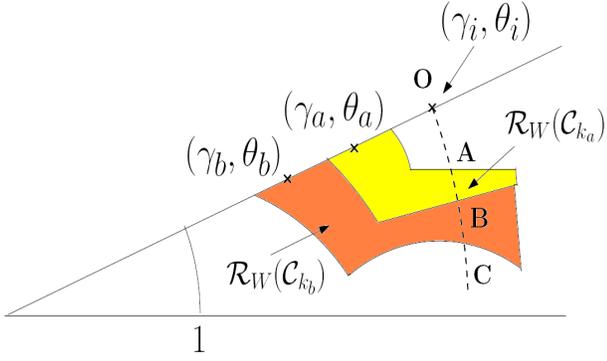}
\label{pointersect}
\caption{Diagram illustrates the variation of the minimum distance in $\mathcal{S}_{\text{eff}}$, for fixed $\gamma_i$ on varying $\theta$}
\end{figure}

The procedure to obtain the optimal phase $ \theta_{i,opt} $ of the fade state, for the violation circle centred at singular fade state $ (\gamma_i, \theta _i) \in \mathcal{H}_W $ is stated as follows:
\begin{itemize}
\item[Step 1]
Find the points of intersections of the arc $ \gamma = \gamma_i $, $ \theta \in [0, \pi/M] $, with the boundaries that satisfy the conditions mentioned in Lemma \ref{lem:optimalboundary}.
\item[Step 2]
If there is only one such point of intersection, say with the boundary between the regions $\mathcal{R}_W(\mathcal{C}_{k_a})$ and $\mathcal{R}_W(\mathcal{C}_{k_b})$,  $ \theta_{i,opt} $ is obtained by solving the equation $ d_{\mathcal{C}_{k_a}}^2(\gamma,\theta)=d_{\mathcal{C}_{k_b}}^2(\gamma,\theta) \vert _{\gamma=\gamma_i} $. On the other hand, if there are $L$ such points of intersections, say with boundaries between regions $\mathcal{R}_W(\mathcal{C}_{k_{a,l}})$ and $\mathcal{R}_W(\mathcal{C}_{k_{b,l}})$, $1\leq l \leq L$, calculate the phase of each of these points of intersection, $ \theta_{l,\text{intersect}} $, $ 1\leq l \leq L $ by solving the equation $ d_{\mathcal{C}_{k_{a,l}}}^2(\gamma,\theta)=d_{\mathcal{C}_{k_{b,l}}}^2(\gamma, \theta) \vert _{\gamma=\gamma_i}$. Then compute  the minimum distances in $ \mathcal{S}_{\text{eff}} $ for the fade state corresponding to the point of intersection $ (\gamma_i , \theta_{l,\text{intersect}}) $ as  $ d_{min}(\gamma_i, \theta_{l,\text{intersect}})= d_{\mathcal{C}_{k_{a,l}}}^2 (\gamma, \theta)\vert_{\gamma=\gamma_i, \theta= \theta_{l,\text{intersect}}} $. Choose 
\begin{align}
\nonumber
l'&= \text{arg}\,\max_{1\leq l \leq L} d_{min}(\gamma_i , \theta_{l,\text{intersect}}).
\end{align}
Then we have, 
\begin{align}
\nonumber
\theta_{i,opt}&= \theta_{l',\text{intersect}}.
\end{align}
\end{itemize}

The optimal rotation angles for the violation circle with centres at $ (\gamma_i,\theta_i) $, can now be easily calculated from $ \theta_{i,opt} $. The optimal rotation angle, $ \alpha_{i,opt} $, for the violation circle centred at $ (\gamma_i, \theta_i) $ is that rotation angle that transforms the fade state from $ (\gamma_i, \theta_i) $ to $ (\gamma_i,\theta_{i,opt})  $. (See Lemma \ref{relative}.)

For violation circles with centre at $ (\gamma_i ,\theta_i) $ , $ 1\leq i \leq N_W $, the optimal rotation angles for the User-2 are as follows:  
\begin{itemize}
\item
If $ \theta_i = \pi/M $, from Lemma \ref{relative}, the optimal rotation is $ \alpha_{i,opt} = \pi/M - \theta_{i,opt} $ in a clockwise direction. 
\item
If $ \theta_i = 0  $, from Lemma \ref{relative}, the optimal rotation is $ \alpha_{\text{i,opt}} = \theta_{\text{i,opt}} $ in an anticlockwise direction.
\end{itemize}
 
\begin{example}\textit{Optimal angles for QPSK signal sets}
When both users use QPSK constellations at input, the channel quantization is shown in Fig. \ref{fig:qpsksector}. The optimal rotation angle are calculated as shown below. These are shown in Fig. \ref{fig:optimal}.

\begin{figure}[t]
\centering
\includegraphics[totalheight=2.5in,width=2.5in]{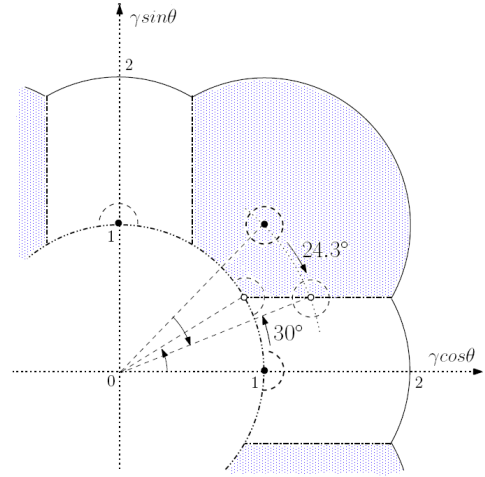}
\caption{Optimal rotation angles for the QPSK case}	
\label{fig:optimal}	
\end{figure}
\begin{itemize}
\item
For violation circle centred at $ (\sqrt{2},\frac{\pi}{4}) $, the optimal phase  $ \theta_{1,opt}  $  is the phase of point of intersection of the arc $ \gamma = \sqrt{2} $ and pairwise boundary $ d_{\mathcal{C}_{k_1}}^2(\gamma,\theta)=d_{\mathcal{C}_{k_2}}^2(\gamma, \theta) $ \emph{i.e.} $ \gamma \sin \theta = 0.5 $. Therefore, $ \theta_{1,opt} = \sin ^{-1} (\frac{1}{2\sqrt{2}}) \approx 20.7 ^\circ$. Thus the optimal rotation angle for the constellation of User-2 relative to User-1 is $ \alpha_{1,opt} = 45^\circ - 20.7^\circ = 24.3 ^\circ $ in a clockwise direction.
\item
For violation circle centred at $ (1,0) $, the optimal phase corresponds to  the point of intersection of the arc $ \gamma =1 $ and boundary $ d_{\mathcal{C}_{k_1}}^2(\gamma, \theta) = d_{\mathcal{C}_{k_2}}^2(\gamma,\theta) $, \emph{i.e.} $ \gamma \sin \theta = 0.5 $. Thus $ \theta _{2,opt} = \sin ^{-1}(0.5) = 30^\circ $. Thus the optimal rotation angle for User-2 is $ \alpha_{2,opt} = 30^\circ $ in an anticlockwise direction relative to User-1. 
\end{itemize}
\end{example}
\begin{example}\textit{Optimal Angles for 8-PSK signal sets}
When both users are using 8-PSK constellations at the input, the channel quantization is obtained as shown in Fig. \ref{fig:8psksector}. The optimal angles of rotation for User-2 relative to User-1, $ \alpha _{i,opt} $ can be calculated using the same technique for all the violation circles centred at the singular fade states $ (\gamma_i,\theta_i) $, $ 1\leq i \leq 7 $. For each $ (\gamma_i, \theta_i) $, $ 1\leq i \leq 7 $, the boundaries corresponding to the optimal phase, $ \theta_{i,opt} $ and $ \alpha _{i,opt} $ are tabulated in Table \ref{table2}. The letters (C) or (A) in the column corresponding to $ \alpha_{i,opt} $ indicates the direction of rotation as clockwise and anticlockwise respectively.
\begin{table*}
\caption{Optimal angles of rotation for 8-PSK case}
\label{table2}
\begin{center}
\begin{tabular}{|c|l|c|l|l|}
\hline
$ i $ & $ (\gamma_i, \theta_i) $ & Boundary & $ \theta _{i,opt} $ & $ \alpha _{i,opt} $ (C/A)  \\ \hline
$ 1 $ & $(1,0)$  & $ d^2_{\mathcal{C}_{k_1}}(\gamma, \theta)= d^2_{\mathcal{C}_{k_4}}(\gamma, \theta) $ & $ \tan ^{-1} \left( \frac{1}{2}\right)- \cos ^{-1} \left( \frac{\gamma^2_1 + \frac{1}{2}+\frac{1}{\sqrt{2}}}{\sqrt{5}\gamma_1}\right) \approx 17.3^\circ $ & $ 17.3  $ (A) \\ \hline
$ 2 $ &  $(\sqrt{2},0)$ & $ d^2_{\mathcal{C}_{k_1}}(\gamma, \theta)= d^2_{\mathcal{C}_{k_5}}(\gamma, \theta) $ & $ \tan ^{-1} \left( \frac{1}{2\sqrt{2}-1}\right)- \cos ^{-1} \left( \frac{\gamma^2_2 + 2}{2 \gamma _2 \sqrt{5- 2 \sqrt{2}}}\right) \approx 12.4^ \circ $  & $ 12.4^\circ $ (A)  \\ \hline
$ 3 $ & $(\sqrt{2}+1 ,0)$ & $ d^2_{\mathcal{C}_{k_3}}(\gamma, \theta)= d^2_{\mathcal{C}_{k_7}}(\gamma, \theta) $ & $ \sin ^{-1}\left(\frac{1}{2 \gamma_3} \right) \approx 12^\circ$ & $ 12^\circ $ (A)\\ \hline
$ 4 $ & $\left(\sqrt{4-2\sqrt{2}},\frac{\pi}{8}\right)$ & $ d^2_{\mathcal{C}_{k_1}}(\gamma, \theta)= d^2_{\mathcal{C}_{k_4}}(\gamma, \theta) $ & $ \tan ^{-1} \left( \frac{1}{2}\right)- \cos ^{-1} \left( \frac{\gamma^2_4 + \frac{1}{2}+\frac{1}{\sqrt{2}}}{\sqrt{5}\gamma_4}\right) \approx 15.9^\circ$ & $ 6.6^\circ $ (C)\\ \hline
$ 5 $ & $\left(\sqrt{1+\frac{1}{\sqrt{2}}} ,\frac{\pi}{8}\right)$ & $ d^2_{\mathcal{C}_{k_1}}(\gamma, \theta)= d^2_{\mathcal{C}_{k_5}}(\gamma, \theta) $ & $ \tan ^{-1} \left( \frac{1}{2\sqrt{2}-1}\right)- \cos ^{-1} \left( \frac{\gamma^2_5 + 2}{2 \gamma_5 \sqrt{5-2\sqrt{2}}}\right) \approx 13^\circ$ & $ 9.5^\circ $  (C)\\ \hline
$ 6 $ & $\left(\sqrt{2+\sqrt{2}} ,\frac{\pi}{8}\right)$ & $ d^2_{\mathcal{C}_{k_3}}(\gamma, \theta)= d^2_{\mathcal{C}_{k_6}}(\gamma, \theta) $ & $ \cos ^{-1} \left(\frac{1+\sqrt{2}}{2\gamma_6}\right)-\frac{\pi}{4} \approx 4.2^\circ$  & $ 18.3^\circ $ (C)\\ \hline
$ 7 $ & $\left(\sqrt{4+2\sqrt{2}},\frac{\pi}{8}\right)$ & $ d^2_{\mathcal{C}_{k_3}}(\gamma, \theta)= d^2_{\mathcal{C}_{k_7}}(\gamma, \theta) $ & $ \sin^{-1} \left(\frac{1}{2\gamma_7} \right) \approx 11^\circ$ & $ 11.5^\circ $ (C)\\ \hline
\end{tabular}
\end{center}
\end{table*}
\end{example}

\subsection{Upper Bound on $\delta$} \label{sec:upper}

In this subsection we obtain an upper bound on $ \delta $ for which the proposed rotation scheme can be employed. It is necessary that the violation circle corresponding to any of the singular fade states must not overlap with any of the effective shifted circles, otherwise the minimum distance guarantee will be violated. This is illustrated in Fig. \ref{violation_delta}. It is clear from the figure, to avoid the overlap it is necessary that the distance between the centre of each of the effective shifted circle $ (\gamma_i,\theta_{i,opt}) $, $ 1 \leq i \leq N_W $ and the singular fade states $ (\gamma_j, \theta_j) $, $ 1\leq j \leq N_W $ should be at least equal to sum of the radius of the shifted circle and the radius of the violation circles centred at $ (\gamma_j,\theta_j) $. It is required, 
\noindent for each $i$, $1\leq i\leq N_W $ ,
\begin{align}
\label{eqn:upperbound}
\rho(\gamma_i,\theta_i)+\rho(\gamma_j,\theta_j) &\leq d_{(\gamma _i, \theta _{i,opt}) \leftrightarrow (\gamma_j, \theta_j)} \\
\nonumber
 &\text{for all }j\,,1\leq j\leq N_W, 
\end{align}  
where $d_{(\gamma _i, \theta _{i,opt}) \leftrightarrow (\gamma_j, \theta_j)}$ is the Euclidean distance between the points $(\gamma_j, \theta_j)$ and $(\gamma_i, \theta_{i,opt})$ in the $(\Gamma, \Theta)$ plane.
Since $ \rho (\gamma_k, \theta_k) $, $ 1\leq k \leq N_W $, is a function of $ \delta $, \eqref{eqn:upperbound} provides an upper bound on $ \delta $. 
\begin{example}
When both users use QPSK constellations at the input, both the violation circles centred at $ (1,0^\circ) $ and $ (\sqrt{2},45^\circ) $ has radius $ \delta /\sqrt{2} $ \emph{i.e.} $\rho(1,0^\circ) = \rho (\sqrt{2},45^\circ) = \frac{\delta}{\sqrt{2}}$. Now,
\begin{align}
\nonumber
&d_{(\sqrt{2},45^\circ)\leftrightarrow (\sqrt{2},20.7^\circ) }=d_{(1,0^\circ)\leftrightarrow (\sqrt{2},20.7^\circ) } \approx 0.5936\\
\nonumber
&d_{(\sqrt{2},45^\circ) \leftrightarrow (1,30^\circ) }= d_{(1,0^\circ) \leftrightarrow (1,30^\circ) } \approx  0.5176.
\end{align}
To avoid overlap, from \eqref{eqn:upperbound}, we have 
\begin{align}
\nonumber
&2(\frac{\delta}{\sqrt{2}}) \leq \min \lbrace 0.5936, 0.5176 \rbrace \\
\nonumber
& \therefore \delta \leq 0.365 \approx \delta_{max}.
\end{align}
\end{example}

For $ \delta > \delta _{max} $ there always exist some fade states for which the minimum distance in $ \mathcal{S}_{\text{eff}} $ cannot be increased beyond $ \delta $ using the proposed scheme. For example, in Fig. \ref{violation_delta}, the fade state corresponding to the point $ P $ is transferred to $ P' $ after rotation. But $ P' $ still lies within the violation circle corresponding to singular fade state $ (\sqrt{2},\frac{\pi}{4}) $, thus minimum distance guarantee is violated.

\begin{figure}[t]
\centering
\includegraphics[totalheight=3in,width=3in]{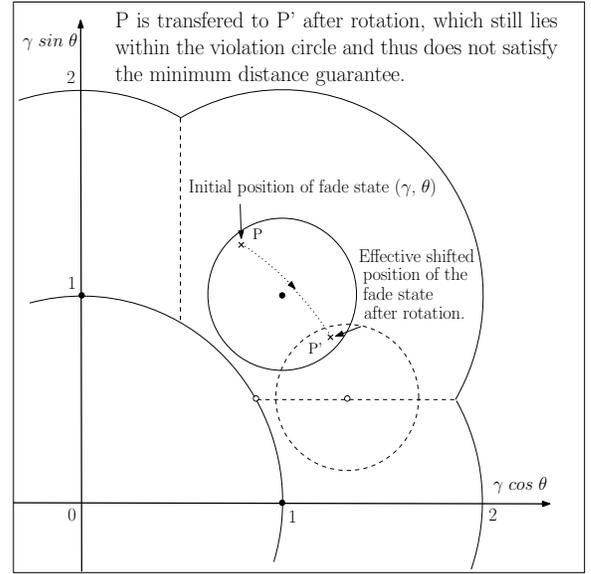}
\caption{Diagram illustrating the necessity for an upper bound on $\delta$ }	
\label{violation_delta}	
\end{figure}

\section{SIMULATION RESULTS AND DISCUSSIONS} \label{sec:results}

All through the previous section we have assumed that the ratio $\frac{h_2}{h_1}$ is calculated at the destination. But even if is not so, \emph{i.e.} actually the ratio $\frac{h_1}{h_2}$ is used for channel quantization at the destination, then exactly the same scheme would work except that instead of rotating the constellation of User-2 we have to rotate the signal set of User-1. The optimal angles of rotation will still be same as calculated before. The feedback that the destination sends back to the users indicates if the fade state $ (\gamma,\theta) $ lies in any of the violation circles or not, and if it does, then identifies in which one of the $ N_W $ violation circles does it lie in. It also needs to indicate which one among the  two ratios is calculated for fade state quantization at the destination. Thus the total feedback overhead is  $\lceil \log_2 (N_W +1)\rceil +1$ bits. For example, the feedback overhead for the QPSK and 8-PSK case are 3 and 4 bits respectively. This feedback overhead is very nominal. 

\begin{figure}[t]
\centering
\includegraphics[totalheight=2.4in,width=3.1in]{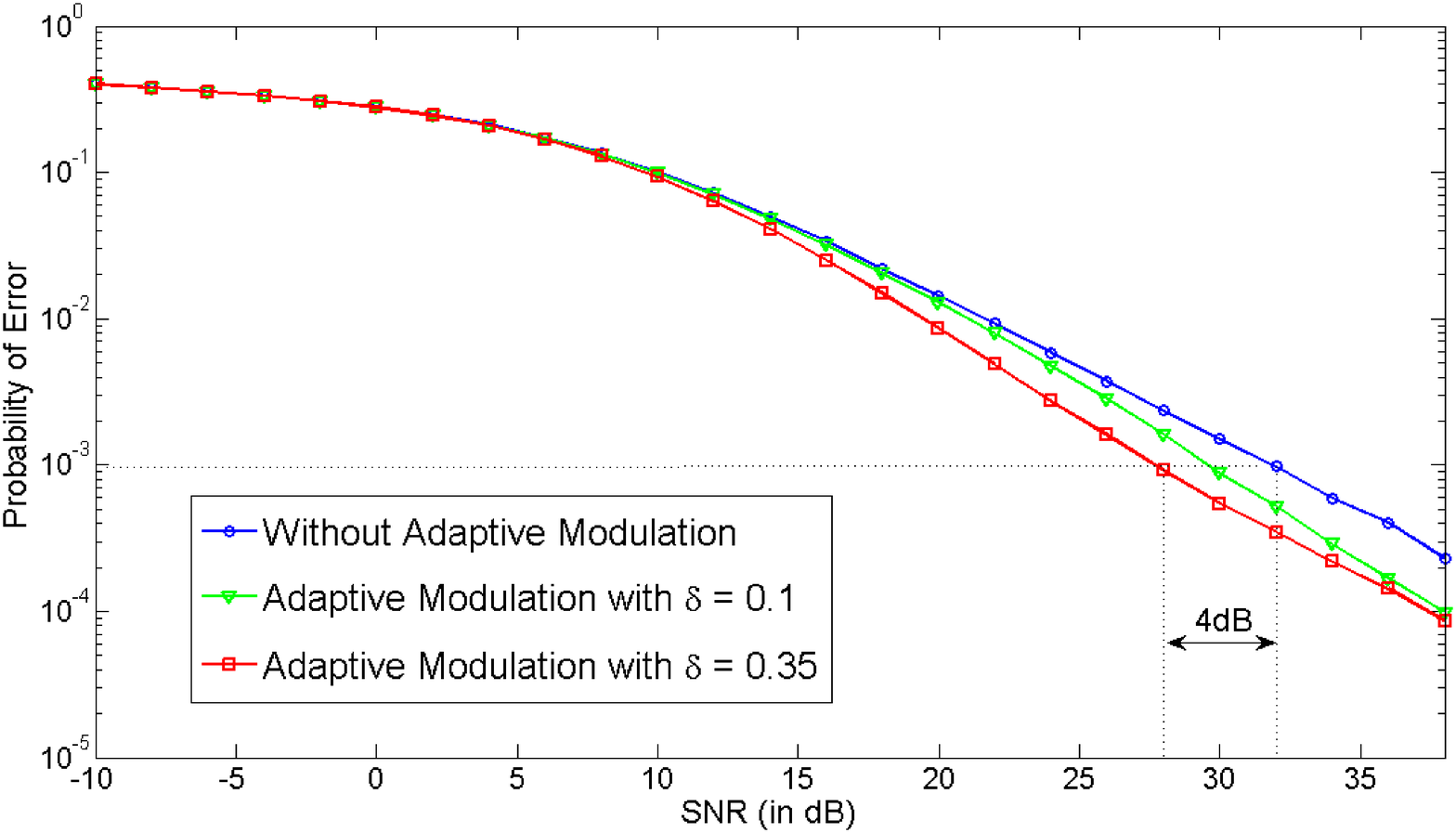}
\caption{Probability of error vs SNR plot, when both users use QPSK signal sets}	
\label{PeSNR}	
\end{figure}

The system is simulated for the case when both users use QPSK signal sets at the input.
The probability of error is  plotted against SNR in Fig. \ref{PeSNR} both without and with adaptive modulation for different values of $ \delta $. The gains obtained increases on increasing $ \delta $ as expected. For a $ P_e= 10^{-3} $ there is a 4dB gain obtained with adaptive modulation with $ \delta=0.35 $, as shown in the figure. 

\section{CONCLUSIONS}

In this paper, we have proposed a modulation scheme for the two-user MAC  with fading which adapts according to the fade states. For this purpose we have obtained a quantization of all possible fade states based on a minimum distance criteria when both users use $M$-PSK constellations at input. We have identified the regions, called violation circles, such that when the fade state lies in them the minimum distance requirement in $\mathcal{S}_{\text{eff}}$  is violated. The quantized fade state knowledge is fed back to the  users using just $\lceil \log_2 (N_W +1)\rceil +1 $ bits. Based on this quantized feedback, one of the users rotates it's constellation to effectively overcome the `bad channel conditions'.  
We have shown the extent to which the error performance of this proposed scheme is better than the conventional scheme without adaptation. The case when both users use QAM constellations at input has not been considered in this paper, is a natural topic for future work. Other cases with more than two users communicating with a single destination is also an interesting direction for future work.
  
%%%%%%%%%%%%%%%%%%%%%%%%%%%%%%%%%%%%%%%%%%%%%%%%%%%%%%%%%%%%%%%%%%%%%%%%%%%%%%%%%%%%%%%%%%
\section*{Acknowledgement}
This work was supported  partly by the DRDO-IISc program on Advanced Research in Mathematical Engineering through a research grant as well as the INAE Chair Professorship grant to B.~S.~Rajan.
%%%%%%%%%%%%%%%%%%%%%%%%%%%%%%%%%%%%%%%%%%%%%%%%%%%%%%%%%%%%%%%%%%%%%%%%%%%%%%%%%%%%%%

\end{document}